\def\final{1}

\documentclass[11pt]{article}
\usepackage{dsfont}
\usepackage{algorithm, algorithmic}
\usepackage{amsmath, amssymb, amsthm}
\usepackage{framed}
\usepackage{fullpage}
\usepackage{graphicx}
\usepackage{verbatim}
\usepackage{float}
\usepackage{color}
\usepackage{times}
\definecolor{DarkGreen}{rgb}{0.1,0.5,0.1}
\definecolor{DarkRed}{rgb}{0.5,0.1,0.1}
\definecolor{DarkBlue}{rgb}{0.1,0.1,0.5}
\usepackage{enumitem}
\usepackage[pdftex]{hyperref}
\hypersetup{
    unicode=false,				% non-Latin characters in Acrobat¿s bookmarks
    pdftoolbar=true,				% show Acrobat toolbar?
    pdfmenubar=true,			% show Acrobat menu?
    pdffitwindow=false, 		% page fit to window when opened
    pdftitle={},    					% title
    pdfauthor={}
    pdfsubject={},				% subject of the document
    pdfnewwindow=true,		% links in new window
    pdfkeywords={},			% list of keywords
    colorlinks=true,				% false: boxed links; true: colored links
    linkcolor=DarkRed,		% color of internal links
    citecolor=DarkGreen,	% color of links to bibliography
    filecolor=DarkRed,		% color of file links
    urlcolor=DarkBlue,		% color of external links
}

\ifnum\final=0
\newcommand{\mynote}[1]{\marginpar{\tiny \sf #1}}
\else
\newcommand{\mynote}[1]{}
\fi
\newcommand{\mnote}[1]{\mynote{Mark: {#1}}}
\newcommand{\knote}[1]{\mynote{Kobbi: {#1}}}

\newcommand{\remove}[1]{}

\newcommand{\Lap}{\operatorname{\rm Lap}}

\newcommand{\AAA}{\mathcal A}

\newcommand{\PPP}{\mathcal P}

\newcommand\N{\mathbb{N}}

\newcommand{\cA}{\mathcal{A}}

\newcommand{\cC}{\mathcal{C}}
\newcommand{\cD}{\mathcal{D}}

\newcommand{\cF}{\mathcal{F}}
\newcommand{\cP}{\mathcal{P}}

\newcommand{\cS}{\mathcal{S}}

\newcommand{\eps}{\epsilon}
\newcommand{\poly}{\mathrm{poly}}

\newcommand{\bits}{\{0,1\}}
\newcommand{\getsr}{\gets_{\mbox{\tiny R}}}

\newcommand{\prob}[1]{\mathrm{Pr}\left[#1\right]}

\newtheorem{theorem}{Theorem}[section]
\newtheorem{lemma}[theorem]{Lemma}

\newtheorem{remark}[theorem]{Remark}
\newtheorem{corollary}[theorem]{Corollary}
\newtheorem{proposition}[theorem]{Proposition}
\newtheorem{observation}[theorem]{Observation}
\newtheorem{example}[theorem]{Example}

\theoremstyle{definition}
\newtheorem{definition}[theorem]{Definition}

\newcommand{\vcdim}{\operatorname{VC}}
\newcommand{\error}{{\rm error}}

\newcommand{\VC}{\operatorname{\rm VC}}

\newcommand{\RepDim}{\operatorname{\rm RepDim}}

\newcommand{\point}{\operatorname{\tt POINT}}

\newcommand{\thresh}{\operatorname*{\tt THRESH}}
\newcommand{\parity}{\operatorname{\tt PAR}}

\newcommand{\OPT}{\operatorname{\rm OPT}}

\def\poly{\mathop{\rm{poly}}\nolimits}

\newcommand{\adist}{\mathcal{A}_{\rm dist}}
\newcommand{\gap}{\operatorname{gap}}

\newcommand{\gen}{\operatorname{Gen}}
\newcommand{\trace}{\operatorname{Trace}}
\newcommand{\fpadv}{\mathcal{A}_{FP}}
\renewcommand{\sec}{\xi}
\newcommand{\codebook}{W}
\newcommand{\codeword}{w}
\newcommand{\pirateword}{w'}

\title{Simultaneous Private Learning of Multiple Concepts}
\author{Mark Bun\thanks{John A. Paulson School of Engineering \& Applied Sciences, Harvard University. Supported by an NDSEG fellowship and NSF grant CNS-1237235. Work done in part while the author was visiting Yale University. \texttt{mbun@seas.harvard.edu}}
\quad 
Kobbi Nissim\thanks{Dept.\ of Computer Science, Ben-Gurion University {\em and} Center for Research on Computation \& Society (CRCS), Harvard University. Supported by NSF grant CNS-1237235, a gift from Google, Inc., a Simons Investigator grant, and ISF grant 276/12. \texttt{kobbi@cs.bgu.ac.il}, \texttt{kobbi@seas.harvard.edu}} 
\quad  
Uri Stemmer\thanks{Dept.\ of Computer Science, Ben-Gurion University. Supported by the Ministry of Science and Technology (Israel), by the IBM PhD Fellowship Awards Program, and by the Frankel Center for Computer Science. \texttt{stemmer@cs.bgu.ac.il}}}

\begin{document}

\maketitle

\begin{abstract}

We investigate the {\em direct-sum} problem in the context of differentially private PAC learning: What is the sample complexity of solving $k$ learning tasks {\em simultaneously} under differential privacy, and how does this cost compare to that of solving $k$ learning tasks without privacy? In our setting, an individual example consists of a domain element $x$ labeled by $k$ unknown concepts $(c_1,\ldots,c_k)$. The goal of a {\em multi}-learner is to output $k$ hypotheses $(h_1,\ldots,h_k)$ that generalize the input examples.

Without concern for privacy, the sample complexity needed to simultaneously learn $k$ concepts is essentially the same as needed for learning a single concept. Under differential privacy, the basic strategy of learning each hypothesis independently yields sample complexity that grows polynomially with $k$. For some concept classes, we give multi-learners that require fewer samples than the basic strategy. Unfortunately, however, we also give lower bounds showing that even for very simple concept classes, the sample cost of private multi-learning must grow polynomially in $k$.

\end{abstract}

\thispagestyle{empty}

\vfill


\newpage

\setcounter{page}{1}

\section{Introduction}

The work on {\em differential privacy}~\cite{DMNS06} is aimed at providing useful analyses on privacy-sensitive data while providing strong individual-level privacy protection. One family of such analyses that has received a lot of attention is PAC learning~\cite{Valiant84}. These tasks abstract many of the computations performed over sensitive information~\cite{KLNRS08}.

We address the {\em direct-sum} problem -- what is the cost of solving multiple instances of a computational task simultaneously as compared to solving each of them separately? -- in the context of differentially private PAC learning. In our setting, individual examples are drawn from domain $X$ and labeled by $k$ unknown concepts $(c_1,\ldots,c_k)$ taken from a concept class $C=\left\{c:X\rightarrow \{0,1\}\right\}$, i.e., each example is of the form
$ (x,y_1,\ldots,y_k)$, where $x\in X$ and $y_i=c_i(x)$.  
The goal of a multi-learner is to output $k$ hypotheses $(h_1,\ldots,h_k)$ that generalize the input examples while preserving the privacy of individuals.

The direct-sum problem has its roots in complexity theory, and is a basic problem for many algorithmic tasks. It also has implications for the practical use of differential privacy. Consider, for instance, a hospital that collects information about its patients and wishes to use this information for medical research. The hospital records for each patient a collection of attributes such as age, sex, and the results of various diagnostic tests (for each patient, these attributes make up a point $x$ in some domain $X$) and, for each of $k$ diseases, whether the patient suffers from the disease (the $k$ labels $(y_1,\ldots,y_k)$). Based on this collection of data, the hospital researchers wish to learn good predictors for the $k$ diseases. One option for the researchers is to perform each of the learning tasks on a fresh sample of patients, hence enlarging the number of patient examples needed (i.e.\ the {\em sample complexity}) by a factor of $k$, which can be very costly. 

Without concern for privacy, the sample complexity that is necessary and sufficient for performing the $k$ learning tasks is actually fully characterized by the VC dimension of the concept class $C$  -- it is independent of the number of learning tasks $k$. In this work, we set out to examine if the situation is similar when the learning is performed with differential privacy. Interestingly, we see that with differential privacy the picture is quite different, and in particular, the required number of examples can grow polynomially in $k$.

\paragraph{Private learning.}

A {\em private learner} is an algorithm that is given an sample of labeled examples $(x,c(x))$ (each representing the information and label pertaining to an individual) and outputs a generalizing hypothesis $h$ that guarantees differential privacy with respect to its examples. The first differentially private learning algorithms were given by Blum et al.~\cite{BDMN05} and the notion of {\em private learning} was put forward and formally researched by Kasiviswanathan et al.~\cite{KLNRS08}. Among other results, the latter work presented a generic construction of differentially private learners with sample complexity $O(\log|C|)$.

In contrast, the sample complexity of (non-private) PAC learning is $\Theta(\VC(C))$, which can be much lower than $\log|C|$ for specific concept classes. This gap led to a line of work examining the sample complexity of private learning, which has revealed a significantly more complex picture than there is for non-private learning. In particular, for {\em pure} differentially private learners, it is known that the sample complexity of proper learning (where the learner returns a hypothesis $h$ taken from $C$) is sometimes higher than the sample complexity of improper learners (where $h$ comes from an arbitrary hypothesis class $H$). The latter is characterized by the {\em representation dimension} of the concept class $C$, which is generally higher than the VC dimension~\cite{BKN10,BBKN14,CH11,BNS13a,FX14}. By contrast, a sample complexity gap between proper and improper learners does not exist for non-private learning. In the case of {\em approximate} differential privacy no such combinatorial characterization is currently known. It is however known that the sample complexity of such learners can be significantly lower than that of pure-differentially private learners and yet higher than the VC dimension of $C$~\cite{BNS13a,FX14,BNSV15}. Furthermore, there exist (infinite) PAC-learnable concept classes for which no differentially private learner (pure or approximate) exists.

\paragraph{Private multi-learning.} In this work we examine the sample complexity of private multi-learning. Our work is motivated by the recurring research theme of the direct-sum, as well as by the need to understand whether multi-learning remains feasible under differential privacy, as it is without privacy constraints.

At first glance, private multi-learning appears to be similar to the query release problem, the goal of which is to approximate the average values of a large collection of predicates on a dataset. One surprising result in differential privacy is that it is possible to answer an exponential number of such queries on a dataset~\cite{BLR08,RR10,HR10}. For example, Blum, Ligett, and Roth~\cite{BLR08} showed  that given a dataset $D$ and a concept class $C$, it is possible to generate with differential privacy a dataset $\hat{D}$ such that the average value of $c$ on $D$ approximates the average of $c$ on $\hat{D}$ for every $c \in C$ {\em simultaneously}.
%$$\frac{|\{i:x_i\in D: c(x_i)=1\}|}{|D|} \approx \frac{|\{i:\hat x_i\in \hat D: c(\hat x_i)=1\}|}{|\hat D|}\quad \mbox{for all}~c\in C.$$ 
The sample complexity required, i.e., the size of the database $D$, to perform this sanitization is only logarithmic in $|C|$. Results of this flavor suggest that we can also learn exponentially many concepts simultaneously. However, we give negative results showing that this is not the case, and that multi-learning can have significantly higher sample complexity than query release.

\subsection{Our results}

Prior work on privately learning the simple concept classes $\point_{X}$ (of functions that evaluate to $1$ on exactly one point of their domain $X$ and to $0$ otherwise) and $\thresh_{X}$ (of functions that evaluate to $1$ on a prefix of the domain $X$ and to $0$ otherwise) has demonstrated a rather complex picture, depending on whether learners are proper or improper, and whether learning is performed with pure or approximate differential privacy~\cite{BKN10,BBKN14,BNS13a,BNS13b,BNSV15}. We analyze the sample complexity of multi-learning of these simple concept classes, as well as general concept classes. We also consider the class $\parity_d$ of parity functions, but in this case we restrict our attention to uniformly selected examples. We examine both proper and improper PAC and agnostic learning under pure and approximate differential privacy. For ease of reference, we include tables with our results in Section~\ref{sec:ourresults}, where we omit the dependency on the privacy and accuracy parameters. 

\paragraph{Techniques for private $k$-learning.}
Composition theorems for differential privacy show that the sample complexity of learning $k$ concepts simultaneously is at most a factor of $k$ larger than the sample complexity of learning one concept (and may be reduced to $\sqrt{k}$ for approximate differential privacy). Unfortunately, privately learning one concept from a concept class $C$ can sometimes be quite costly, requiring much higher sample complexity than $\VC(C)$ which is needed to learn non-privately. Building on techniques of Beimel, Nissim, and Stemmer \cite{BNS15}, we show that the multiplicative dependence on $k$ can always be reduced to the VC-dimension of $C$, at the expense of producing a one-time sanitization of the dataset.

\begin{theorem}[Informal] \label{thm:generic-informal}
Let $C$ be a concept class for which there is pure differentially private sanitizer for $C^\oplus = \{f \oplus g : f, g \in C\}$ with sample complexity $m$. Then there is an pure differentially private agnostic $k$-learner for $C$ with sample complexity $O(m + k \cdot \VC(C))$.

Similarly, if $C^\oplus$ has an approximate differentially private sanitizer with sample complexity $m$, then there is an approximate differentially private agnostic $k$-learner for $C$ with sample complexity $O(m + \sqrt{k} \cdot \VC(C))$.
\end{theorem}

The best known general-purpose sanitizers require sample complexity $m = O(\VC(C) \log |X|)$ for pure differential privacy \cite{BLR08} and $m = O(\log |C| \sqrt{\log |X|})$ for approximate differential privacy \cite{HR10}. However, for specific concept classes (such as $\point_{X}$ and $\thresh_{X}$), the sample complexity of sanitization can be much lower.

In the case of approximate differential privacy, the sample complexity of $k$-learning can be even lower than what is achievable with our generic learner. Using stability-based arguments, we show that point functions and parities under the uniform distribution can be PAC $k$-learned with sample complexity $O(\VC(C))$ -- independent of the number of concepts $k$ (see Theorems \ref{thm:pointsUpper} and \ref{thm:paritiesUniformUpper}).

\paragraph{Lower bounds.}
In light of the above results, one might hope to be able to reduce the dependence on $k$ further, or to eliminate it entirely (as is possible in the case of non-private learning). We show that this is not possible, even for the simplest of concept classes. In the case of pure differential privacy, a packing argument \cite{FFKN09, HT10, BKN10} shows that any non-trivial concept class requires sample complexity $\Omega(k)$ to privately $k$-learn (Theorem \ref{thm:OmegaK}). For approximate differential privacy, we use fingerprinting codes \cite{BS98, BUV14} to show that unlike points and parities, threshold functions require sample complexity $\tilde{\Omega}(k^{1/3})$ to PAC learn privately (Corollary \ref{cor:thresholds}). Moreover, any non-trivial concept class requires sample complexity $\tilde{\Omega}(\sqrt{k})$ to privately learn in the agnostic model (Theorem \ref{thm:fpclb-agnostic}). In the case of point functions, this matches the upper bound achievable by our generic learner.

\medskip
\noindent We highlight a few of the main takeaways from our results:
\smallskip

\paragraph{A complex answer to the direct sum question.} Our upper bounds show that solving $k$ learning problems simultaneously  can require substantially lower sample complexity than solving the problems individually. On the other hand, our lower bounds show that a significant dependence on $k$ is generally necessary. 
\paragraph{Separation between private PAC and private agnostic learning.} Non-privately, the sample complexities of PAC and agnostic learning are of the same order (differing only in the dependency in the accuracy parameters). 
Beimel et al.~\cite{BNS15} showed that this is also the case with differentially private learning (of one concept).
Our results on learning point functions show that private PAC and agnostic {\em multi-learning} can be substantially different (even for learning up to constant error). In the case of approximate differential privacy, $O(1)$ sample suffice to PAC-learn multiple point functions. However, $\tilde{\Omega}(\sqrt{k})$ samples are needed to learn $k$ points agnostically. 
\paragraph{Separation between improper learning with approximate differential privacy and non-private learning.} Recently, Bun et al. \cite{BNSV15} showed that the sample complexity of learning one threshold function with approximate differential privacy exceeds the VC dimension, but only in the case of \emph{proper} learning. Thus it remains possible that improper learning with approximate differential privacy can match the sample complexity of non-private learning. While we do not address this question directly, we exhibit a separation for multi-learning. In particular, learning $k$ thresholds with approximate differential privacy requires $\tilde{\Omega}(k^{1/3})$ samples, even improperly, while $O(1)$ samples suffices non-privately.

%\unote{Maybe emphasize the fact that currently no separation is known between non-private learning and improper learning with $(\eps,\delta)$-privacy. In the $k$-learning model we can prove a separation.}

%\mnote{General $C$ lower bounds are witnessed by specific lower bounds elsewhere.}

\subsection{Related work}

Differential privacy was defined in~\cite{DMNS06} and the relaxation to approximate differential privacy is from~\cite{DKMMN06}. Most related to our work is the work on private learning and its sample complexity~\cite{BDMN05,KLNRS08,CH11,DRV10,BBKN14,BNS13a,BNS13b,FX14,BNS15,BNSV15} and the early work on sanitization~\cite{BLR08}. 
That many ``natural'' learning tasks can be performed privately was shown in the early work of Blum et al.~\cite{BDMN05} and Kasiviswanathan et al.~\cite{KLNRS08}. A characterization for the sample complexity of {\em pure-private} learners was given in~\cite{BNS13a}, in terms of a new combinatorial measure -- the {\em Representation Dimension}, that is, given a class $C$, the number of samples needed and sufficient for privately learning $C$ is $\Theta(\RepDim(C))$. Building on~\cite{BNS13a}, Feldman and Xiao~\cite{FX14} showed an equivalence between the representation dimension of a concept $C$ and the randomized one-way communication complexity of the evaluation problem for concepts from $C$. Using this equivalence they separated the sample complexity of pure-private learners from that of non-private ones.

\remove{
Dwork et al.~\cite{DRV10} showed how to boost the accuracy of private learning algorithms. That is, given a {\em private} learning algorithm that
has a big classification error, they produced a {\em private} learning algorithm with small error.
Other tools for private learning include, e.g., private logistic regression~\cite{CM08}, and private empirical risk minimization~\cite{CMS11}. \knote{not sure we need this paragraph}
}

The problem of learning multiple concepts simultaneously (without privacy) has been considered before. Motivated by the problem of bridging computational learning and reasoning, Valiant~\cite{Valiant06} also observed that (without privacy) multiple concepts can be learned from a common dataset in a data efficient manner.

%\knote{Valiant: http://people.seas.harvard.edu/~valiant/AAAI06.pdf}

\subsection{Tables of results}
\label{sec:ourresults}

The following tables summarize the results of this work. In the tables below $C$ is a class of concepts (i.e., predicates) defined over domain $X$. 
Sample complexity upper and lower bounds is given in terms of $|C|$ and $|X|$. Note that for $\point_{X}$, $\thresh_{X}$, and $\parity_d$ we have $|C|=\Theta(|X|)$.

Where not explicitly noted, upper bounds hold for the setting of agnostic learning and lower bounds are for the (potentially easier) setting of PAC learning. Similarly, where not explicitly noted, upper bounds are for proper learning and lower bounds are for the (less restrictive) setting of improper learning. For simplicity, these tables hide constant and logarithmic factors, as well as dependencies on the learning and privacy parameters.

\paragraph{Multi-learning with pure differential privacy.}\ \\

\smallskip

\noindent Upper bounds:
\begin{center}
\begin{tabular}{|c | c | c | c | c | c |}
\hline
& \multicolumn{2}{|c|}{PAC learning} & \multicolumn{2}{|c|}{Agnostic learning} & \\
$C$ & proper & improper & proper & improper & References\\ 
\hline
$\point_{X}$ & $k + \log |C|$  & $k$ & $k + \log |C|$ &  $k$ & Thm. \ref{thm:GeneralUpper}, Cor. \ref{cor:pureImproperPoint}\\ 
\hline
$\thresh_{X}$ & \multicolumn{4}{|c|}{$k + \log |C|$}  & Thm. \ref{thm:GeneralUpper}\\ 
\hline
General & \multicolumn{4}{|c|}{$\min\{k \log|C|, k \VC(C) + \log|X| \VC(C)\}$} & Thm. \ref{thm:GeneralUpper}\\ 
\hline\hline
$\parity_{d}$ (uniform)	& \multicolumn{4}{|c|}{$k \log |C|$} & Thm. \ref{thm:GeneralUpper}\\ 
\hline
\end{tabular}
\end{center}
\noindent
Lower bounds:

\begin{center}
\begin{tabular}{|c | c | c | c | c | c |}
\hline
& \multicolumn{2}{|c|}{PAC learning} & \multicolumn{2}{|c|}{Agnostic learning} & \\
$C$ & proper & improper & proper & improper & References \\ 
\hline
$\point_{X}$ & $k + \log |C|$ & $k$ & $k + \log |C|$ &  $k$ & Thm. \ref{thm:OmegaK}, \cite{BKN10}\\ 
\hline
$\thresh_{X}$ & \multicolumn{4}{|c|}{$k + \log |C|$} & Thm. \ref{thm:OmegaK}, \cite{BKN10, FX14}  \\ 
\hline\hline
$\parity_{d}$ (uniform)	& \multicolumn{4}{|c|}{$k \log |C|$} & Thm. \ref{thm:paritiesLowerBound} \\ 
\hline
\end{tabular}
\end{center}

\paragraph{Multi-learning with approximate differential privacy.} \ \\

\smallskip

\noindent 
Upper bounds:
\begin{center}
\begin{tabular}{|c | c | c |}
\hline
 & PAC learning & Agnostic  learning \\
$C$ & (proper and improper) & (proper and improper) \\
\hline
$\point_{X}$ & $1$ \; (Thm. \ref{thm:pointsUpper}) & $\sqrt{k}$ \; (Cor. \ref{cor:pointsAgnosticUpper})\\ 
\hline
$\thresh_{X}$ & \multicolumn{2}{|c|}{$2^{\log^*|X|} + \sqrt{k}$  \; (Cor. \ref{cor:thresholdsAgnosticUpper})} \\
\hline
General $C$ & \multicolumn{2}{|c|}{$\min\{\sqrt{k} \log|C|,  \sqrt{k} \VC(C) + \log|X| \VC(C), \sqrt{k} \VC(C) + \sqrt{\log|X|} \log|C|\}$ \; (Thm. \ref{thm:GeneralUpper})} \\ 
\hline\hline
$\parity_{d}$ (uniform)	& $\log |C|$ \; (Thm. \ref{thm:paritiesUniformUpper}) & $\sqrt{k}\log |C|$ \; (Thm. \ref{thm:GeneralUpper})\\ 
\hline
\end{tabular}
\end{center}

\newpage

\noindent
Lower bounds:
\begin{center}
\begin{tabular}{|c | c | c | c | c | c |}
\hline
& \multicolumn{2}{|c|}{PAC learning} & \multicolumn{2}{|c|}{Agnostic learning} &  \\
$C$ & proper & improper & proper & improper & References \\ 
\hline
$\point_{X}$ &\multicolumn{2}{|c|}{1} & \multicolumn{2}{|c|}{$\sqrt{k}$} & Cor. \ref{cor:OmegaSqrtK}\\ 
\hline
$\thresh_{X}$ & $\log^*|X| + k^{1/3}$ & $k^{1/3}$ &  $\log^*|X| + \sqrt{k}$ & $\sqrt{k}$ & Cor. \ref{cor:thresholds}, Cor. \ref{cor:OmegaSqrtK}, \cite{BNSV15} \\ 
\hline\hline
$\parity_{d}$ (uniform)	&\multicolumn{2}{|c|}{$\log|C|$}  & \multicolumn{2}{|c|}{$\sqrt{k} + \log|C|$} & Cor. \ref{cor:OmegaSqrtK}) \\ 
\hline
\end{tabular}
\end{center}

\section{Preliminaries}
\label{sec:preliminaries}

We recall and extend standard definitions from learning theory and differential privacy.

\subsection{Multi-learners}

In the following $X$ is some arbitrary domain. A concept (similarly, hypothesis) over domain $X$ is a predicate defined over $X$. A concept class (similarly, hypothesis class) is a set of concepts. 

\begin{definition}[Generalization Error]
Let $\cP\in \Delta(X\times\{0,1\})$ be a probability distribution over $X\times\{0,1\}$.
The {\em generalization error} of a hypothesis $h:X\rightarrow\{0,1\}$ w.r.t.\ $\cP$ is defined as 
$\error_{\cP}(h)=\Pr_{(x,y) \sim \cP}[h(x)\neq y]$.

Let $\cD\in \Delta(X)$ be a probability distribution over $X$ and let $c:x\rightarrow\{0,1\}$ be a concept. 
The {\em generalization error} of hypothesis $h:X\rightarrow\{0,1\}$ w.r.t.\ $c$ and $\cD$ is  
defined as $\error_{\cD}(c,h)=\Pr_{x \sim \cD}[h(x)\neq c(x)]$.
If $\error_{\cD}(c,h)\leq\alpha$ we say that $h$ is {\em $\alpha$-good} for $c$ and $\cD$.
\end{definition}

\begin{definition}[Multi-labeled database] A {\em $k$-labeled} database over a domain $X$ is a database $S\in(X\times\{0,1\}^k)^*$. That is, $S$ contains $|S|$ elements from $X$, each concatenated with $k$ binary labels.
\end{definition}

%\begin{definition}[Empirical error]
%Given an unlabeled database $D=(x_i)_{i=1}^n$ and two concepts $c,h$, we define the {\em empirical error} of $h$ w.r.t. $c$ and $D$ as $\error_{D}(h,c)=\frac{1}{n}|\{i : h(x_i)\neq c(x_i)\}|$.
%
%Given a $k$-labeled database $S=(x_i,y_{i,1},\dots,y_{i,k})_{i=1}^n$ and a hypothesis $h$, we define the {\em empirical error} of $h$ w.r.t.\ the $j^{\text{th}}$ label column in $S$ as
%$\error_{S|_j}(h)=\frac{1}{n}|\{i : h(x_i)\neq y_{i,j}\}|$.
%\end{definition}

Let $\cA:\left(X\times\{0,1\}^k\right)^n\rightarrow \left(2^X\right)^k$ be an algorithm that operates on a $k$-labeled database and returns $k$ hypotheses.
Let $C$ be a concept class over a domain $X$ and let $H$ be a hypothesis class over $X$. 
We now give a generalization of the notion of PAC learning~\cite{Valiant84} to multi-labeled databases (the standard PAC definition is obtained by setting $k=1$):

%\paragraph{Definition~\ref{def:pacmultilearner}} (PAC Multi-Learner). \;
\begin{definition}[PAC Multi-Learner]\label{def:pacmultilearner}
Algorithm $\cA$ is an $(\alpha,\beta)$-PAC $k$-learner for concept class $C$ using hypothesis class $H$ with sample complexity $n$ if for every distribution $\cD$ over $X$ and for every fixture of $(c_1,\dots,c_k)$ from $C$, given a $k$-labeled database as an input $S=\left(\left(x_i,c_1(x_i),\dots,c_k(x_i)\right)\right)_{i=1}^n$ where each $x_i$ is drawn i.i.d.\ from $\cD$, algorithm $\cA$ outputs $k$ hypotheses $(h_1,\dots,h_k)$ from $H$ satisfying
$$\Pr\left[\max_{1\leq j \leq k} \left(\error_{\cD}(c_j,h_j)\right)  > \alpha\right] \leq \beta.$$
The probability is taken over the random choice of
the examples in $S$ according to $\cD$ and the coin tosses of the learner $\cA$.
If $H \subseteq C$ then $A$ is called a {\em proper} learner; otherwise, it is called an {\em improper} learner.
\end{definition}

\begin{definition}[Agnostic PAC Multi-Learner] \label{def:agnosticmultilearner}
Algorithm $\cA$ is an $(\alpha,\beta)$-PAC agnostic $k$-learner for $C$ using hypothesis class $H$ and sample complexity $n$ if for every distribution $\cP$ over $X\times\{0,1\}^k$, given a $k$-labeled database $S=\left(\left(x_i,y_{1,i},\dots,y_{k,i}\right)\right)_{i=1}^n$ where each $k$-labeled sample $(x_i,y_{1,i}\ldots,y_{k,i})$ is drawn i.i.d.\ from $\cP$, algorithm $\cA$ outputs $k$ hypotheses $(h_1,\dots,h_k)$ from $H$ satisfying
$$\Pr\left[\max_{1\leq j\leq k}\left(\error_{\cP_j}(h_j)  -  \min_{c \in C}\left(\error_{\cP_j}(c)\right) \right)> \alpha \right] \leq \beta,$$
where $\cP_j$ is the marginal distribution of $\cP$ on the examples and the $j^{\text{th}}$ label. The probability is taken over the random choice of
the examples in $S$ according to $\cP$ and the coin tosses of the learner $\cA$.
If $H\subseteq C$ then $A$ is called a {\em proper} learner; otherwise, it is called an {\em improper} learner.
\end{definition}

%\knote{define VC dimension. Theorem about sample complexity of multi-learners.}

\subsection{The Sample Complexity of Multi-Learning}

Without privacy considerations, the sample complexities of PAC and agnostic learning are essentially characterized by a combinatorial quantity called the \emph{Vapnik-Chervonenkis (VC) dimension}. We state these characterizations in the context of multi-learning.

\subsubsection{The Vapnik-Chervonenkis Dimension}\label{sec:VC}

\begin{definition}
Fix a concept class $C$ over domain $X$. A set $\{x_1, \dots, x_d\} \in X$ is \emph{shattered} by $C$ if for every labeling $b \in \{0, 1\}^d$, there exists $c \in C$ such that $b_1 = c(x_1), \dots, b_d = c(x_d)$. The \emph{Vapnik-Chervonenkis (VC) dimension} of $C$, denoted $\vcdim(C)$, is the size of the largest set which is shattered by $C$.
\end{definition}

The Vapnik-Chervonenkis (VC) dimension is an important combinatorial measure of a concept class.
Classical results in statistical learning theory show that the generalization error of a hypothesis $h$ and its empirical error (observed on a large enough sample) are similar.

\begin{definition}[Empirical Error] 
Let $S = ((x_i, y_i))_{i = 1}^n \in (X\times\{0,1\})^n$ be a labeled sample from $X$.
The {\em empirical error} of a hypothesis $h:X\rightarrow\{0,1\}$ w.r.t.\ $S$ is defined as 
$\error_{S}(h)= \frac{1}{n} |\{i : h(x_i)\neq y_i\}|$.

Let $D \in X^n$ be a (unlabeled) sample from $X$ and let $c:x\rightarrow\{0,1\}$ be a concept. 
The {\em empirical error} of hypothesis $h:X\rightarrow\{0,1\}$ w.r.t.\ $c$ and $D$ is  
defined as $\error_{D}(c,h)= \frac{1}{n} |\{i : h(x_i)\neq c(x_i)]$.
\end{definition}

\begin{theorem}[VC-Dimension Generalization Bound, e.g. \cite{BlumerEhHaWa89}]\label{thm:Generalization}
Let $\cD$ and $C$ be, respectively, a distribution and a concept class over a domain $X$, and let $c \in C$. For a sample $S=((x_i,c(x_i)))_{i=1}^n$ where $n\geq\frac{64}{\alpha}(\VC(C)\ln(\frac{64}{\alpha})+\ln(\frac{8}{\beta}))$
and the $x_i$ are drawn i.i.d. from $\cD$, it holds that
\[\Pr\Big[\exists h\in C \text{ s.t.\ } \error_\cD(h,c)> \alpha \ \land \ \error_S(h)\leq\frac{\alpha}{2} \Big] \le \beta.\]
%$$\Pr\Big[\forall \; h\in H:\;\; \big|\error_\cD(h,c)-\error_S(h)\big|\leq\alpha\Big]\geq1-\beta.$$  
\end{theorem}

This generalization argument extends to the setting of \emph{agnostic learning}, where a hypothesis with small empirical error might not exist.

\begin{theorem}[VC-Dimension Agnostic Generalization Bound, e.g. \cite{AB09,AnthonySh93}]\label{thm:AgnosticGeneralization}
Let $H$ be a concept class over a domain $X$, and let $\PPP$ be a distribution over $X\times\{0,1\}$. For a sample $S=((x_i,y_i))_{i=1}^n$ containing $n\geq\frac{64}{\alpha^2}(\VC(H)\ln(\frac{6}{\alpha})+\ln(\frac{8}{\beta}))$ i.i.d.\ elements from $\PPP$, it holds that
$$\Pr\Big[\exists h\in H \text{ s.t.\ } \big|\error_\PPP(h)-\error_S(h)\big|>\alpha\Big]\leq\beta.$$  
\end{theorem}

Using theorems~\ref{thm:generalization} and~\ref{thm:AgnosticGeneralization}, an upper bound of $O(\vcdim(C))$ on the sample complexity of learning a concept class $C$ follows by reduction to the \emph{empirical learning} problem. The goal of empirical learning is similar to that of PAC learning, except accuracy is measured only with respect to a fixed input database. Theorems~\ref{thm:generalization} and~\ref{thm:AgnosticGeneralization} state that when an empirical learner is run on sufficiently many samples, it is also accurate with respect to a distribution on inputs.

\begin{definition}[Empirical Learner]
Algorithm $\cA$ is an {\em $(\alpha,\beta)$-accurate empirical $k$-learner} for a concept class $C$ using hypothesis class $H$ with sample complexity $n$ if for every collection of concepts $(c_1, \dots, c_k)$ from $C$ and database $S =((x_i, c_1(x_i), \dots, c_k(x_i)))_{i=1}^n\in(X\times\{0,1\}^k)^n$, algorithm $\cA$ outputs $k$ hypotheses $(h_1, \dots, h_k)$ from $H$ satisfying
$$\Pr\left[\max_{1\leq j \leq k} \left(\error_{S|_j}(h_j)\right)  > \alpha\right] \leq \beta,$$
where $S|_j=((x_i, c_j(x_i)))_{i = 1}^n$. The probability is taken over the coin tosses of $\cA$.
\end{definition}

\begin{definition}[Agnostic Empirical Learner]
Algorithm $\cA$ is an {\em agnostic $(\alpha,\beta)$-accurate empirical $k$-learner} for a concept class $C$ using hypothesis class $H$ with sample complexity $n$ if for every database $S =((x_i, y_{1, i}, \dots, y_{k, i}))_{i=1}^n\in(X\times\{0,1\}^k)^n$, algorithm $\cA$ outputs $k$ hypotheses $(h_1, \dots, h_k)$ from $H$ satisfying
$$\Pr\left[\max_{1\leq j\leq k}\left(\error_{S|_j}(h_j)  -  \min_{c \in C}\left(\error_{S|_j}(c)\right) \right)> \alpha \right] \leq \beta,$$
where $S|_j=((x_i, y_{j, i}))_{i = 1}^n$. The probability is taken over the coin tosses of $\cA$.
\end{definition}

\begin{theorem} \label{thm:generalization}
Let $\cA$ be an $(\alpha,\beta)$-accurate empirical $k$-learner for a concept class $C$ (resp. agnostic empirical $k$-learner) using hypothesis class $H$. Then $\cA$ is also a $(2\alpha, \beta + \beta')$-accurate PAC learner for $C$ when given at least $\max \{n,  \frac{32}{\alpha}(\VC(H \oplus C) \log(32/\alpha) + \log(8/\beta'))\}$ samples (resp. $\max \{n, \frac{64}{\alpha^2}(\VC(H) \log(6/\alpha) + \log(8k/\beta'))$ samples). Here, $H \oplus C = \{h \oplus c : h \in H, c \in C\}$.
\end{theorem}

\begin{proof}
We begin with the non-agnostic case. Let $\cA$ be an $(\alpha,\beta)$-accurate empirical $k$-learner for $C$. Let $\cD$ be a distribution over the example space $X$. Let $S$ be a random i.i.d. sample of size $m$ from $\cD$. The generalization bound for PAC learning (Theorem \ref{thm:Generalization}) states that if $m \ge  \frac{32}{\alpha}(d \log(32/\alpha) + \log(8/\beta'))\}$, then
\[\Pr[\exists c \in C, h\in H: \error_S(c, h) \le \alpha \land \error_{\cD}(c, h) > 2\alpha] \le \beta',\]
where $d = \VC(H \oplus C)$. The result follows by a union bound over the failure probability of $\cA$ and the failure of generalization.

Now we turn to the agnostic case. Let $\cA$ be an agnostic $(\alpha,\beta)$-accurate empirical $k$-learner for $C$. Fix an index $j \in [k]$, and let $\cP_j$ be a distribution over $X \times \{0, 1\}$. Let $S$ be a random i.i.d. sample of size $m$ from $\cP_j$. Then generalization for agnostic learning (Theorem \ref{thm:AgnosticGeneralization}) yields
\[\Pr[\exists h \in H : |\error_{S}(h) - \error_{\cP_j}(h)| > \alpha] \le \frac{\beta'}{k}\]
for $m \ge \frac{64}{\alpha^2}(\VC(H) \log(6/\alpha) + \log(8k/\beta'))\}$.The result follows by a union bound over the failure probability of $\cA$ and the failure of generalization for each of the indices $j = 1, \dots, k$.
\end{proof}

Applying the above theorem in the special case where $\cA$ finds the concept $c \in C$ that minimizes the empirical error on its given sample, we obtain the following sample complexity upper bound for proper multi-learning.

\begin{corollary} \label{coro:multi-learningVC}
%\paragraph{Corollary~\ref{coro:multi-learningVC}.} {\em 
Let $C$ be a concept class with VC dimension $d$. There exists an $(\alpha, \beta)$-accurate proper PAC $k$-learner for $C$ using $O(\frac{1}{\alpha}(d \log (1/\alpha) + \log(1/\beta))$ samples. Moreover, there exists an $(\alpha, \beta)$-accurate proper agnostic PAC $k$-learner for $C$ using $O(\frac{1}{\alpha^2}(d \log (1/\alpha) + \log(k/\beta))$ samples.
%}
\end{corollary}

\begin{proof}
For the non-agnostic case, we simply let $\cA$ be the $(0, 0)$-accurate empirical learner that outputs any vector of hypotheses that is consistent with its given examples (one is guaranteed to exist, since the target concept satisfies this condition). The claim follows from Theorem \ref{thm:generalization} noting that $\VC(C \oplus C) = O(\VC(C))$.

For the agnostic case, consider the algorithm $\cA$ that on input $S$ outputs hypotheses $(h_1, \dots, h_k)$ that minimize the quantities $\error_{S_j}(h_j)$. Applying the agnostic generalization bound \cite{AB09}, this is an $(\alpha / 2, \beta / 2)$-accurate agnostic empirical learner given $O(\frac{1}{\alpha^2}(d \log (1/\alpha) + \log(k/\beta))$ samples. The claim then follows from Theorem \ref{thm:generalization}.
\end{proof}

It is known that even for $k = 1$, the sample complexities of PAC and agnostic learning are at least $\Omega(\VC(C)/\alpha)$ and $\Omega(\VC(C)/\alpha^2)$, respectively. Therefore, the above sample complexity upper bound is tight up to logarithmic factors.

%\knote{define concept classes POINT, THRESH, PAR}

We define a few specific concept classes which will play an important role in this work.

\begin{description}

\item[$\point_X$:] Let $X$ be any domain. The class of \emph{point functions} is the set of all concepts that evaluate to $1$ on exactly one element of $X$, i.e. $\point_X = \{c_x : x \in X\}$ where $c_x(y) = 1$ iff $y = x$. The VC-dimension of $\point_X$ is $1$ for any $X$.

\item[$\thresh_X$:] Let $X$ be any totally ordered domain. The class of \emph{threshold functions} takes the form $\thresh_X = \{c_x : x \in X\}$ where $c_x(y) = 1 \text{ iff } y \le x$. The VC-dimension of $\thresh_X$ is $1$ for any $X$.

\item[$\parity_d$:] Let $X = \bits^d$. The class of \emph{parity functions} on $X$ is given by $\parity_d = \{c_x : x \in X\}$ where $c_x(y) = \langle x, y \rangle \pmod 2$. The VC-dimension of $\parity_d$ is $d$.

\end{description}

In this work, we focus our study of the concept class $\parity_d$ on the problem of learning parities under the uniform distribution. The PAC and agnostic learning problems are defined as before, except we only require a learner to be accurate when the marginal distribution on examples is the uniform distribution $U_d$ over $\{0, 1\}^d$. 

\begin{definition}[PAC Learning $\parity_d$ under Uniform]
Algorithm $\cA$ is an $(\alpha,\beta)$-PAC $k$-learner for $\parity_d$ using hypothesis class $H$ and sample complexity $n$ if for every fixed $(c_1,\dots,c_k)$ from $C$, given a $k$-labeled database as an input $S=\left(\left(x_i,c_1(x_i),\dots,c_k(x_i)\right)\right)_{i=1}^n$ where each $x_i$ is drawn i.i.d.\ from $U_d$, algorithm $\cA$ outputs $k$ hypotheses $(h_1,\dots,h_k)$ from $H$ satisfying
$$\Pr\left[\max_{1\leq j \leq k} \left(\error_{U_d}(c_j,h_j)\right)  > \alpha\right] \leq \beta.$$
\end{definition}

\begin{definition}[Agnostically Learning $\parity_d$ under Uniform]
Algorithm $\cA$ is an $(\alpha,\beta)$-PAC agnostic $k$-learner for $\parity_d$ using hypothesis class $H$ and sample complexity $n$ if for every distribution $\cP$ over $\bits^d \times\{0,1\}^k$, with marginal distribution $U_d$ over the data universe $\bits^d$, given a $k$-labeled database $S=\left(\left(x_i,y_{1,i},\dots,y_{k,i}\right)\right)_{i=1}^n$ where each $k$-labeled sample $(x_i,y_{1,i}\ldots,y_{k,i})$ is drawn i.i.d.\ from $\cP$, algorithm $\cA$ outputs $k$ hypotheses $(h_1,\dots,h_k)$ from $H$ satisfying
$$\Pr\left[\max_{1\leq j\leq k}\left(\error_{\cP_j}(h_j)  -  \min_{c \in C}\left(\error_{\cP_j}(c)\right) \right)> \alpha \right] \leq \beta,$$
where $\cP_j$ is the marginal distribution of $\cP$ on the examples and the $j^{\text{th}}$ label.
\end{definition}

\subsection{Differential privacy}

Two  {\em $k$-labeled} databases $S,S' \in(X\times\{0,1\}^k)^n$ are called {\em neighboring} if they differ on a single (multi-labeled) entry, i.e., $|\{i:(x_i,y_{1,i},\ldots,y_{k,i}) \not= (x'_i, y'_{1,i},\ldots,y'_{k,i})\}| = 1$.

%\paragraph{Definition~\ref{def:differentialprivacy}} (Differential Privacy~\cite{DMNS06}).\;
\begin{definition}[Differential Privacy~\cite{DMNS06}] 
Let $\cA: \left(X\times\{0,1\}^k\right)^n\rightarrow \left(2^X\right)^k$ be an algorithm that operates on a $k$-labeled database and returns $k$ hypotheses. Let $\epsilon,\delta \geq 0$. Algorithm $\cA$ is $(\epsilon,\delta)$-differentially private if for all neighboring $S,S'$ and for all $T\subseteq \left(2^X\right)^k$,
$$\Pr[\cA(S) \in T] \leq e^\epsilon \cdot \Pr[\cA(S') \in T] + \delta,$$ 
where the probability is taken over the coin tosses of the algorithm $\cA$. When $\delta=0$ we say that $\cA$ satisfies {\em pure} differential privacy, otherwise (i.e., if $\delta >0$) we say that $\cA$ satisfies {\em approximate} differential privacy.
\end{definition}

Our learning algorithms are designed via repeated applications of differentially private algorithms on a database. Composition theorems for differential privacy show that the price of privacy for multiple (adaptively chosen) interactions degrades gracefully.

\begin{theorem}[Composition of Differential Privacy \cite{DKMMN06, DL09, DRV10}]\label{thm:composition}
Let $0 < \eps, \delta' < 1$ and $\delta \in [0, 1]$. Suppose an algorithm $\cA$ accesses its input database $S$ only through $m$ adaptively chosen executions of $(\eps, \delta)$-differentially private algorithms. Then $\cA$ is 
\begin{enumerate}
\item $(m \eps, m\delta)$-differentially private, and
\item $(\eps', m \delta + \delta')$-differentially private for $\eps = \sqrt{2m\ln(1/\delta')} \cdot \eps + 2m\eps^2$.
\end{enumerate}
\end{theorem}

\subsection{Differentially Private Sanitization}

A fundamental task in differential privacy is the \emph{data sanitization} problem. Given a database $D = (x_1, \dots, x_n) \in X^n$, the goal of a sanitizer is to privately produce a synthetic database $\hat{D} \in X^m$ that captures the statistical properties of $D$. We are primarily interested in sanitization for boolean-valued functions (equivalently referred to as \emph{counting queries}). Given a function $c: X \to \{0, 1\}$ and a database $D = (x_1, \dots, x_n)$, we write $c(D) = \frac{1}{n} \sum_{i = 1}^n c(x_i)$.

\begin{definition}[Sanitization]
An algorithm $\cA : X^n \to X^m$ is an $(\alpha, \beta)$-accurate sanitizer for a concept class $C$ if for every $D \in X^n$, the algorithm $\cA$ produces a database $\hat{D} \in X^m$ such that
\[\Pr[\exists c \in C : |c(D) - c(\hat{D}) | > \alpha] \le \beta.\]
Here, the probability is taken over the coins of $\cA$.
\end{definition}

In an influential result, Blum, Ligett, and Roth \cite{BLR08} showed that any concept class $C$ admits a differentially private sanitizer with sample complexity $O(\VC(C) \log |X|)$:

\begin{theorem}[\cite{BLR08}]\label{thm:BLR}
For any concept class $C$ over a domain $X$, there exists an $(\alpha, \beta)$-accurate and $(\eps, 0)$-differentially private sanitizer $\cA: X^n \to X^m$ for $C$ when
\[n = O\left(\frac{\VC(C) \cdot \log|X| \cdot \log(1/\alpha)}{\alpha^3 \eps} + \frac{\log(1/\beta)}{\alpha \eps}\right),\]
and $m = O(\VC(C) \log(1/\alpha) / \alpha^2)$.
\end{theorem}

When relaxing to $(\eps, \delta)$-differential privacy, the private multiplicative weights algorithm of Hardt and Rothblum \cite{HR10} can sometimes achieve lower sample complexity (roughly $O(\log|C| \sqrt{\log |X|})$).

\begin{theorem}[\cite{HR10}]\label{thm:HR}
For any concept class $C$ over a domain $X$, there exists an $(\alpha, \beta)$-accurate and $(\eps, \delta)$-differentially private sanitizer $\cA: X^n \to X^m$ for $C$ when
\[n = O\left(\frac{(\log|C| + \log(1/\beta)) \cdot \sqrt{\log|X| \cdot \log(1/\delta)}}{\alpha^2 \eps}\right),\]
and $m = O(\VC(C) \log(1/\alpha) / \alpha^2)$.
\end{theorem}

However, for specific concept classes, sanitizers are known to exist with much lower sample complexity. We first give a sanitizer for point functions with essentially optimal sample complexity, which improves and simplifies a result of \cite{BNS13b}.

\begin{proposition}\label{prop:sanPoint}
There exists an $(\alpha, \beta)$-accurate and $(\eps, \delta)$-differentially private sanitizer for $\point_X$ with sample complexity
\[n = O\left( \frac{\log(1/\alpha\beta\delta)}{\alpha\eps}\right).\]
\end{proposition}

\begin{proof}

To give a $(2\alpha, \beta)$-accurate sanitizer, it suffices to produces, for each point function $c_x$, an approximate answer $a_x \in [0, 1]$ with $|a_x - c_x| \le \alpha$. This is because given these approximate answers, one can reconstruct a database $\hat{D}$ of size $O(1/\alpha)$ with $|c_x(\hat{D}) - a_x| \le \alpha$ for every $x \in X$.

The algorithm for producing the answers $a_x$ is as follows.
\begin{algorithm}[H]
\caption{Query release for $\point_X$}\label{alg:point}
\textbf{Input:} Privacy parameters $(\eps, \delta)$, database $D \in X^n$ \\
For each $x \in X$, do the following:
\begin{enumerate}
\item If $c_x(D) \le \frac{\alpha}{4}$, release $a_x = 0$
\item Let $\hat{a}_x = c_x(D) + \Lap(2/\eps n)$
\item If $\hat{a}_x \le \frac{\alpha}{2}$, release $a_x = 0$
\item Otherwise, release $a_x = \hat{a}_x$
\end{enumerate}
\end{algorithm}

\newcommand{\relpoint}{a}

First, we argue that Algorithm \ref{alg:point} is $(\eps, \delta)$-differentially private. Below, we write $X \approx_{(\eps, \delta)} Y$ to denote the fact that for every measurable set $S$ in the union of the supports of $X$ and $Y$, we have $\Pr[X \in S] \le e^\eps \Pr[Y \in S] + \delta$.

Let $D \sim D'$ be adjacent databases of size $n$, with $x \in D$ replaced by $x' \in D'$. Then the output distribution of the mechanism differs only on its answers to the queries $c_x$ and $c_{x'}$. Let us focus on $c_x$. If both $c_x(D) \le \alpha / 4$ and $c_x(D') \le \alpha/4$, then the mechanism always releases $0$ for both queries. If both $c_x(D) > \alpha/4$ and $c_x(D') > \alpha/4$, then $a_x(D) \approx_{(\eps/2, 0)} a_x(D')$ by properties of the Laplace mechanism. Finally, if $c_x(D) > \alpha / 4$ but $c_x(D') \le \alpha/4$, then $c_x(D') = 0$ with probability $1$. Moreover, we must have $\point_x(D) \le \alpha/4 + 1/n$, so
\[\Pr[a_x(D) = 0] \ge \Pr[Lap(2/\eps n) < \alpha/4 -1/n] = 1 - \frac{1}{2} \exp(-\eps n \alpha / 8 + \eps /2) \ge 1 - \delta/2.\]
So in this case, $a_x(D) \approx_{(0, \delta/2)} a_x(D')$. Therefore, we conclude that overall $a_x(D) \approx_{(\eps/2, \delta/2)} a_x(D')$. An identical argument holds for $a_{x'}$, so the mechanism is $(\eps, \delta)$-differentially private.

Now we argue that the answers $a_x$ are accurate. First, the answers are trivially $\alpha$-accurate for all queries $c_x$ on which $c_x(D) \le \alpha / 4$. For each of the remaining queries, it is $\alpha$-accurate with probability at least
\[\Pr[|Lap(2/\eps n)| < \alpha/2] = 1 - \exp(-\eps n \alpha / 4) \ge 1 - \frac{\alpha\beta}{4}.\]
Taking a union bound over the at most $4/\alpha$ queries with $\point_x(D) > \alpha/4$, we conclude that the mechanism is $\alpha$-accurate for all queries with probability at least $1- \beta$.

\end{proof}

Bun et al. \cite{BNSV15}, improving on work of Beimel et al. \cite{BNS13b}, gave a sanitizer for threshold functions with sample complexity roughly $2^{\log^*|X|}$.

\begin{proposition}[\cite{BNSV15}] \label{prop:sanThreshold}
There exists an $(\alpha, \beta)$-accurate and $(\eps, \delta)$-differentially private sanitizer for $\thresh_X$ with sample complexity
\[n = O\left( \frac{1}{\alpha \eps} \cdot  2^{\log^*|X|} \cdot \log^*|X| \cdot \log \left(\frac{\log^*|X|}{\eps\delta}\right) \cdot \log(1/\beta) \cdot \log^{2.5}(1/\alpha)\right).\]
\end{proposition}

%\knote{cite composition}

%\knote{define sanitization, exponential mechanism, cite BLR08}

%\mnote{cite what is known about private learning for $k = 1$, e.g. transformation from BNS SODA paper}

\subsection{Private learners and multi-learners}

Generalizing on the concept of private learners~\cite{KLNRS08}, we say that an algorithm $\cA$ is $(\alpha,\beta,\epsilon,\delta)$-private PAC $k$-learner for $C$ using $H$ if $\cA$ is $(\alpha,\beta)$-PAC $k$-learner for $C$ using $H$, and  $\cA$ is $(\epsilon,\delta)$-differentially private (similarly with agnostic private PAC $k$-learners). We omit the parameter $k$ when $k=1$ and the parameter $\delta$ when $\delta=0$.

%\knote{recall known results for case $k=1$: general construction, POINT, THRESH, PAR (proper/improper and pure/approx upper/lower bounds)}

For the case $k=1$, we have a generic construction with sample complexity proportional to $\log|C|$:
%\paragraph{Theorem~\ref{thm:klnrs}}(\cite{KLNRS08}).\;
\begin{theorem}[\cite{KLNRS08}]\label{thm:klnrs} 
Let $C$ be a concept class, and $\alpha,\beta,\epsilon > 0$. There exists an $(\alpha,\beta,\epsilon)$-private agnostic proper learner for $C$ with sample complexity $O\left((\log |C| + \log 1/\beta)(1/(\epsilon\alpha) + 1/\alpha^2)\right)$.
\end{theorem}

Beimel, Nissim, and Stemmer \cite{BNS15} gave a generic transformation from data sanitization to private learning, which generally gives improved sample complexity upper bounds.

\begin{theorem}[\cite{BNS15}]\label{thm:BNS15}
Suppose there exists an $(\alpha, \beta)$-accurate and $(\eps, \delta)$-differentially private sanitizer for $C^\oplus$ with sample complexity $m$. Then there exists a proper $(2\alpha, 2\beta)$-PAC and $(\eps + \eps', \delta)$-differentially private learner for $C$ with sample complexity
\[O\left(m + \frac{\VC(C)}{\alpha^3 \eps'} \log \left(\frac{1}{\alpha}\right) + \frac{1}{\alpha \eps'} \log \left(\frac{1}{\beta}\right)\right).\]
\end{theorem}

A number of works \cite{BKN10, BNS13a, BNS13b, FX14, BNSV15} have established sharper upper and lower bounds for learning the specific concept classes $\point_X$ and $\thresh_X$. In the case of pure differential privacy, $\point_X$ requires $\Theta(\log |X|)$ samples to learn properly \cite{BKN10}, but can be learned improperly with $O(1)$ samples. On the other hand, the class of threshold functions $\thresh_X$ require $\Omega(\log |X|)$ samples to learn, even improperly \cite{FX14}. In the case of approximate differential privacy, $\point_x$ and $\thresh_x$ can be learned properly with sample complexities $O(1)$ \cite{BNS13b} and $\tilde{O}(2^{\log^*|X|})$ \cite{BNSV15}, respectively. Moreover, properly learning threshold functions requires sample complexity $\Omega(\log^*|X|)$.

\subsection{Private PAC learning vs. Empirical Learning}

We saw by Theorem \ref{thm:generalization} that when an empirical $k$-learner $\cA$ for a concept class $C$ is run on a random sample of size $\Omega(\VC(C))$, it is also a (agnostic) PAC $k$-learner. In particular, if an empirical $k$-learner $\cA$ is differentially private, then it also serves as a differentially private (agnostic) PAC $k$-learner.

Generalizing a result of \cite{BNSV15}, the next theorem shows that the converse is true as well: a differentially private (agnostic) PAC $k$-learner yields a private empirical $k$-learner with only a constant factor increase in the sample complexity.

%\paragraph{Theorem~\ref{thm:pac-to-empirical}.} 
\begin{theorem}\label{thm:pac-to-empirical}
Let $\eps \le 1$. Suppose $\cA$ is an $(\epsilon,\delta)$-differentially private $(\alpha, \beta)$-accurate (agnostic) PAC $k$-learner for a concept class $\cC$ with sample complexity $n$. Then there is an $(\epsilon,\delta)$-differentially private $(\alpha, \beta)$-accurate (agnostic) empirical $k$-learner $\tilde{\cA}$ for $\cC$ with sample complexity $m = 9n$. Moreover, if $\cA$ is proper, then so is the resulting empirical learner $\tilde{\cA}$.
\end{theorem}

\begin{proof}
We give the proof for the agnostic case; the non-agnostic case is argued identically, and is immediate from \cite{BNSV15}. To construct the empirical learner $\tilde{\cA}$, we use the fact that the given learner $\cA$ performs well on any distribution over labeled examples -- in particular, it performs well on the uniform distribution over rows of the input database to $\tilde{\cA}$. Consider a database $S = ((x_i, y_{1, i}, \dots, y_{k, i}))_{i = 1}^m \in (X \times \{0, 1\}^k)^m$. On input $S$, define $\tilde{\cA}$ by sampling $n$ rows from $S$ (with replacement), and outputting the result of running $\cA$ on the sample. Let $\cS$ denote the uniform distribution over the rows of $S$, and let $\cS_j$ be its marginal distribution which is uniform over $S|_j = ((x_i, y_{j, i}))_{i = 1}^m$. Then sampling $n$ rows from $S$ is equivalent to sampling $n$ rows i.i.d. from $\cS$. Hence, if $(h_1, \dots, h_k)$ is the output of $\cA$ on the subsample, we have
$$\Pr\left[\max_{1\leq j\leq k}\left(\error_{S|_j}(h_j)  -  \min_{c \in C}\left(\error_{S_j}(c)\right) \right)> \alpha \right]  = \Pr\left[\max_{1\leq j\leq k}\left(\error_{\cS_j}(h_j)  -  \min_{c \in C}\left(\error_{\cS_j}(c)\right) \right)> \alpha \right] \leq \beta.$$

To show that $\tilde{\cA}$ remains $(\eps, \delta)$-differentially private, we apply the following ``secrecy-of-the-sample'' lemma \cite{KLNRS08, BNSV15}, which shows that the sampling procedure does not hurt privacy.

\begin{lemma}
Fix $\eps \le 1$ and let $\cA$ be an $(\eps, \delta)$-differentially private algorithm with sample complexity $n$. For $m \ge 2n$, the algorithm $\tilde{\cA}$ described above is $(\tilde{\eps}, \tilde{\delta})$ for
\[\tilde{\eps} = \frac{6\eps m}{n} \quad \text{and} \quad \tilde{\delta} = 4\exp\left(\frac{6\eps m}{n}\right) \cdot \frac{m}{n} \cdot \delta.\]
\end{lemma}
\end{proof}

\section{Upper Bounds on the Sample Complexity of Private Multi-Learners}

\subsection{Generic Construction}

In this section we present the following general upper bounds on the sample complexity of private $k$-learners. 

\begin{theorem}\label{thm:GeneralUpper}
Let $C$ be a finite concept class, and let $k\geq 1$. There exists a proper agnostic $(\alpha,\beta,\epsilon)$-private PAC $k$-learner for $C$ with sample complexity 
\vspace{-7px}
$$O_{\alpha,\beta,\epsilon}\Big(  k\cdot\log k + \min\big\{ k\cdot\log|C| \;,\; (k+ \log|X| )\cdot\VC(C) \big\} \Big),\vspace{-7px}$$
and there exists a proper agnostic $(\alpha,\beta,\epsilon,\delta)$-private PAC $k$-learner for $C$ with sample complexity 
\vspace{-7px}
$$O_{\alpha,\beta,\epsilon,\delta}\left(\sqrt{k}\cdot\log k + \min\left\{ \sqrt{k}\cdot\log|C| \;,\; (\sqrt{k}+\log|X|)\cdot\VC(C) \;,\; \sqrt{k}\cdot\VC(C) + \sqrt{\log|X|}\cdot\log|C| \right\} \right).$$
\end{theorem}

%This yields our proper pure privacy learners (agnostic or not) for $\point_X, \thresh_X, \parity_X$, all with sample complexity $k \log|X|$.\\

The straightforward approach for constructing a private $k$-learner for a class $C$ is to separately apply a (standard) private learner for $C$ for each of the $k$ target concepts. Using composition theorem~\ref{thm:composition} to argue the overall privacy guarantee of the resulting learner, we get the following observation.

\begin{observation}\label{obs:directSum}
Let $C$ be a concept class and let $k\geq1$. If there is an $(\alpha,\beta,\epsilon,\delta)$-PAC learner for $C$ with sample complexity $n$, then 
\vspace{-5px}
\begin{itemize}\setlength\itemsep{-5px}
	\item There is an $(\alpha,k\beta,k\epsilon,k\delta)$-PAC $k$-learner for $C$ with sample complexity $n$.
	\item There is an $(\alpha,k\beta,O(\sqrt{k\log(\frac{1}{\delta})}\epsilon+k\epsilon^2),O(k\delta))$-PAC $k$-learner for $C$ with sample complexity $n$.
\end{itemize}
\vspace{-5px}
Moreover, if the initial learner is proper and/or agnostic, then so is the resulting learner.
\end{observation}

In cases where sample efficient private PAC learners exist, it might be useful to apply Observation~\ref{obs:directSum} in order to obtain a private $k$-learner. For example, Beimel et al.~\cite{BKN10,BNS13a} gave an improper agnostic $(\alpha,\beta,\epsilon)$-PAC learner for $\point_X$ with sample complexity $O_{\alpha}(\frac{1}{\epsilon}\log\frac{1}{\beta})$. Using Observation~\ref{obs:directSum} yields the following corollary.

\begin{corollary}\label{cor:pureImproperPoint}
There exists an improper agnostic $(\alpha,\beta,\epsilon)$-PAC $k$-learner for $\point_X$ with sample complexity $O_{\alpha,\beta,\epsilon}(k\log k)$.
\end{corollary}

For a general concept class $C$, we can use Observation~\ref{obs:directSum} with the generic construction of Theorem~\ref{thm:klnrs}, stating that for every concept class $C$ there exists a private agnostic proper learner $\cal A$ that uses $O(\log|C|)$ labeled examples.

\begin{corollary}\label{cor:straightforward}
Let $C$ be a concept class, and $\alpha,\beta,\epsilon > 0$. 
There exists an $(\alpha,\beta,\epsilon)$-private agnostic proper $k$-learner for $C$ with sample complexity 
$O_{\alpha,\beta,\epsilon}(k\cdot\log|C| + k\cdot\log k)$.
%$O\left((\log |C| + \log\frac{k}{\beta})(\frac{k}{\epsilon\alpha} + \frac{1}{\alpha^2})\right)$.
Moreover, there exists an $(\alpha,\beta,\epsilon,\delta)$-private agnostic proper $k$-learner for $C$ with sample complexity 
$O_{\alpha,\beta,\epsilon,\delta}(\sqrt{k}\cdot\log|C| + \sqrt{k}\cdot\log k)$.
%$O\left((\log |C| + \log\frac{k}{\beta})(\frac{\sqrt{k}}{\epsilon\alpha}\sqrt{\log(1/\delta)} + \frac{1}{\alpha^2})\right)$.
\end{corollary}

\begin{example}\label{eg:PureParityUpper}
There exists a proper agnostic $(\alpha,\beta,\epsilon)$-PAC $k$-learner for $\parity_d$ with sample complexity $O_{\alpha,\beta,\epsilon}(kd+k\log k)$.
\end{example}

As we will see in Section~\ref{sec:pureLowerBounds}, the bounds of Corollary~\ref{cor:pureImproperPoint} and Example~\ref{eg:PureParityUpper} on the sample complexity of $k$-learning $\point_X$ and $\parity_d$ are tight (up to logarithmic factors). That is, with pure-differential privacy, the direct sum gives (roughly) optimal bounds for improperly learning $\point_X$, and for (properly or improperly) learning $\parity_d$. This is not the case for learning $\thresh_X$ or for {\em properly} learning learning $\point_X$.

In order to avoid the factor $k\log|C|$ (or $\sqrt{k}\log|C|$) in Corollary~\ref{cor:straightforward}, we now show how an idea used in~\cite{BNS15} (in the context of semi-supervised learning) can be used to construct sample efficient private $k$-learners. In particular, this construction will achieve tight bounds for learning $\thresh_X$ and for properly learning learning $\point_X$ under pure-differential privacy.

Fix a concept class $C$, target concepts $c_1,\dots,c_k\in C$, and a $k$-labeled database $S$ (we use $D$ to denote the unlabeled portion of $S$).
For every $1\leq j\leq k$, the goal is to identify a hypothesis $h_j\in C$ with low $\error_D(c_j,h_j)$ (such a hypothesis also has good generalization). 
Beimel et al.~\cite{BNS15} observed that given a sanitization $\hat{D}$ of $D$ w.r.t.\ $C^{\oplus}=\{ f{\oplus}g : f,g\in C \}$, for every $f,g\in C$ it holds that 
$$\error_D(f,g)=\frac{1}{|D|}|\{ x\in D : (f\oplus g)(x)=1 \}|\approx\frac{1}{|\hat{D}|}|\{ x\in \hat{D} : (f\oplus g)(x)=1 \}|=\error_{\hat{D}}(f,g).$$

Hence, a hypothesis $h$ with low $\error_{\hat{D}}(h,c_j)$ also has low $\error_D(h,c_j)$ and vice versa. Let $H$ be the set of all dichotomies over $\hat{D}$ realized by $C$. Note that $\exists f_j^*\in H$ that agrees with $c_j$ on $\hat{D}$, i.e., $\exists f_j^*\in H$ s.t.\ $\error_{\hat{D}}(f_j^*,c_j)=0$, and hence $\error_D(f_j^*,c_j)$ is also low. The thing that works in our favor here is that $H$ is small -- at most $2^{|\hat{D}|}\leq 2^{\VC(C)}$ -- and hence choosing a hypothesis out of $H$ is easy.
Therefore, for every $j$ we can use the exponential mechanism to identify a hypothesis $h_j\in H$ with low $\error_D(h_j,c_j)$.

\remove{
\begin{definition}
Given two concepts $h,f\in C$, we denote $(h {\oplus} f): X_d \rightarrow \{0,1\} $, where $(h {\oplus} f)(x)=1$ if and only if $h(x)\neq f(x)$. Let $C^{\oplus}=\{ (h {\oplus} f) \; : \; h,f\in C \}.$
\end{definition}

\begin{observation}\label{obs:vcdim}
For any concept class $C$ it holds that $\VC(C^{\oplus})=O(\VC(C))$.
\end{observation}

\begin{lemma}\label{lem:GenericLearner}
Let $\epsilon'>0$ and let $\AAA$ be an ($\frac{\alpha}{5},\frac{\beta}{5}$)-accurate $(\epsilon,\delta)$-private sanitizer for $C^{\oplus}$ with sample complexity $m$. 
Then there is an $(\alpha,\beta)$-PAC agnostic $k$-learner for $C$ with sample complexity
$$O\left(m+\frac{\VC(C)}{\alpha^3\epsilon'}\log(\frac{1}{\alpha}) +\frac{1}{\alpha\epsilon'}\log(\frac{k}{\beta}) + \frac{1}{\alpha^2}\VC(C)\log(\frac{k}{\alpha\beta})\right).$$
Moreover, it is both $(\epsilon+k\epsilon',\delta)$ and $(\epsilon+\sqrt{2k\ln(1/\delta)}\epsilon'+2k\epsilon'^2,2\delta)$-differentially private.
\end{lemma}

Using Lemma~\ref{lem:GenericLearner} with the generic sanitizer of Theorem~\ref{thm:BLR} results in the following lemma.
}

\begin{lemma}\label{lem:usingSan}
Let $C$ be a concept class, and $\alpha,\beta,\epsilon,\delta > 0$. 
There exists an $(\alpha,\beta,\epsilon)$-private agnostic $k$-learner for $C$ with sample complexity 
%$$O\left(  \frac{\VC(C)}{\alpha^3 \epsilon}\log(\frac{1}{\alpha})(\log|X|+k) + \frac{k}{\alpha \epsilon}\log(\frac{k}{\beta})+\frac{\VC(C)}{\alpha^2}\log(\frac{k}{\alpha\beta})\right).$$
$O_{\alpha,\beta,\epsilon}(\VC(C)\cdot\log|X| + k\cdot\VC(C) + k\cdot\log k)$.
Moreover, there exists an $(\alpha,\beta,\epsilon,\delta)$-private agnostic $k$-learner for $C$ with sample complexity 
$O_{\alpha,\beta,\epsilon,\delta}(\min\{\VC(C)\cdot\log|X|, \log|C| \cdot \sqrt{\log |X|}\} + \sqrt{k}\cdot\VC(C) + \sqrt{k}\cdot\log k)$.
%$$O\left(  \frac{\VC(C)}{\alpha^3 \epsilon}\log(\frac{1}{\alpha})\left(\log|X|+\sqrt{k\log(\frac{1}{\delta})}\right) + \frac{\sqrt{k\log(\frac{1}{\delta})}}{\alpha \epsilon}\log(\frac{k}{\beta})+\frac{\VC(C)}{\alpha^2}\log(\frac{k}{\alpha\beta})\right).$$
%
%$$O\left(  \frac{1}{\alpha^2 \epsilon}\log(\frac{|C|}{\beta})\log(\frac{1}{\delta})\sqrt{\log|X|} +\sqrt{k\log(\frac{1}{\delta})}(\frac{1}{\alpha^3\epsilon}\VC(C)\log(\frac{1}{\alpha}) + \frac{1}{\alpha\epsilon}\log(\frac{k}{\beta})  )+\frac{1}{\alpha^2}\VC(C)\log(\frac{k}{\alpha\beta}) \right).$$
\end{lemma}

Lemma~\ref{lem:usingSan} follows from the following lemma.

\begin{lemma}\label{lem:GenericLearner}
Let $\epsilon'>0$ and let $\AAA$ be an ($\frac{\alpha}{5},\frac{\beta}{5}$)-accurate $(\epsilon,\delta)$-private sanitizer for $C^{\oplus}$ with sample complexity $m$. 
Then there is an $(\alpha,\beta)$-PAC agnostic $k$-learner for $C$ with sample complexity
$$O\left(m+\frac{\VC(C)}{\alpha^3\epsilon'}\log(\frac{1}{\alpha}) +\frac{1}{\alpha\epsilon'}\log(\frac{k}{\beta}) + \frac{1}{\alpha^2}\VC(C)\log(\frac{k}{\alpha\beta})\right).$$
Moreover, it is both $(\epsilon+k\epsilon',\delta)$ and $(\epsilon+\sqrt{2k\ln(1/\delta)}\epsilon'+2k\epsilon'^2,2\delta)$-differentially private.
\end{lemma}

Using Lemma~\ref{lem:GenericLearner} with the generic sanitizer of Theorem~\ref{thm:BLR} or Theorem~\ref{thm:HR} results in Lemma~\ref{lem:usingSan}.

\medskip

An important building block of our generic learner is the exponential mechanism of McSherry and Talwar~\cite{MT07}. A quality function  $q:X^*\times \cF \rightarrow \N$ defines an {\em optimization problem} over the domain $X$ and a finite solution set $\cF$:  Given a database $S \in X^*$, choose $f\in\cF$ that (approximately) maximizes $q(S,f)$. The exponential mechanism solves such an optimization problem sampling a random $f \in \cF$ with probability $\propto \exp\left(\epsilon \cdot q(S,f) /2 \Delta q\right)$. Here, the sensitivity of a quality function, $\Delta q$, is the maximum over all $f\in\cF$ of the sensitivity of the function $q(\cdot, f)$.

\begin{proposition}[Properties of the Exponential Mechanism \cite{MT07}] \label{prop:exp_mech}
\ \begin{enumerate}
\item The exponential mechanism is $(\eps, 0)$-differentially private.
\item Let $q$ be a quality function with sensitivity at most $1$. Fix a database $S \in X^n$ and let $\OPT = \max_{f\in \cF}\{q(S,f)\}$. Let $t >0$. Then exponential mechanism outputs a solution $f$ with $q(S,f)\leq \OPT - tn$ with probability at most $|\cF| \cdot \exp(-\eps tn /2)$.
\end{enumerate}
\end{proposition}

\begin{algorithm}
\caption{$GenericLearner$}\label{alg:genericPrivate}
{\bf Input:} Concept class $C$, privacy parameters $\epsilon',\epsilon,\delta$, and a $k$-labeled database $S=(x_i,y_{i,1},\dots,y_{i,k})_{i=1}^n$. We use $D=(x_i)_{i=1}^n$ to denote the unlabeled portion of $S$.\\
{\bf Used Algorithm:} An ($\frac{\alpha}{5},\frac{\beta}{5}$)-accurate $(\epsilon,\delta)$-private sanitizer for $C^{\oplus}$ with sample complexity $m$.
\begin{enumerate}[rightmargin=10pt,itemsep=1pt]

\item Initialize $H=\emptyset$.

\item Construct an $(\epsilon,\delta)$-private sanitization $\widetilde{D}$ of $D$ w.r.t.\ $C^{\oplus}$, where $|\widetilde{D}|=O\left( \frac{\VC(C^{\oplus})}{\alpha^2}\log(\frac{1}{\alpha}) \right) = O\left( \frac{\VC(C)}{\alpha^2}\log(\frac{1}{\alpha}) \right)$.

\item Let $B=\{b_1,\ldots,b_{|B|}\}$ be the set of all points appearing at least once in $\widetilde{D}$.

\item For every $(z_1,\ldots,z_{|B|})\in \Pi_C(B)=\{\left( c(b_1),\ldots,c(b_{|B|}) \right) :c\in C\}$, add to $H$ an arbitrary concept $c\in C$ s.t.\ $c(b_\ell)=z_\ell$ for every $1\leq \ell\leq |B|$.

\item For every $1\leq j\leq k$, use the exponential mechanism with privacy parameter $\epsilon'$ to choose and return a hypothesis $h_j\in H$ with (approximately) minimal error on the examples in $S$ w.r.t.\ their $j^{\text{th}}$ label.
\end{enumerate}
\end{algorithm}

\begin{proof}[Proof of Lemma~\ref{lem:GenericLearner}]
The proof is via the construction of $GenericLearner$ (algorithm~\ref{alg:genericPrivate}).
Note that $GenericLearner$ only accesses $S$ via a sanitizer (on Step~2) and using the exponential mechanism (on Step~5). 
Composition theorem~\ref{thm:composition} state that $GenericLearner$ is both $(\epsilon+k\epsilon',\delta)$-differentially private and $(\epsilon+\sqrt{2k\ln(1/\delta)}\epsilon'+2k\epsilon'^2,2\delta)$-differentially private.
We, thus, only need to prove that with high probability the learner returns $\alpha$-good hypotheses.

Fix a distribution $\PPP$ over $X\times\{0,1\}^k$, and let $\PPP_j$ denote the marginal distribution of $\PPP$ on the examples and the $j^{\text{th}}$ label. 
Let $S$ consist of examples $(x_i,y_{i,1},\dots,y_{i,k})\sim\PPP$. We use $D=(x_i)_{i=1}^n$ to denote the unlabeled portion of $S$, and use $S|_j=((x_i, y_{j, i}))_{i = 1}^n$ to denote a database containing the examples in $S$ together with their $j^{\text{th}}$ label.
Define the following three events:
\begin{enumerate}[label=$E_{\arabic*}:$]

\item For every $f,h\in C$ it holds that $|\error_D(f,h)-\error_{\widetilde{D}}(f,h)|\leq\frac{2\alpha}{5}$.

\item For every $f\in C$ and for every $1\leq j \leq k$ it holds that $|\error_{S|_j}(f) - \error_{\PPP_j}(f)|\leq\frac{\alpha}{5}$.

\item For every $1\leq j\leq k$, the hypothesis $h_j$ chosen by the exponential mechanism is such that $\error_{S|_j}(h_j) \leq \frac{\alpha}{5} + \min_{f\in H}\left\{\error_{S|_j}(f)\right\}$.

\end{enumerate}

We first argue that when these three events happen algorithm $GenericLearner$ returns good hypotheses.
Fix $1\leq j\leq k$, and let $c_j^* = {\rm argmin}_{f\in C}\{ \error_{\PPP_j}(f) \}$. We denote $\Delta=\error_{\PPP_j}(c_j^*)$.
We need to show that if $E_1 \cap E_2 \cap E_3$ occurs, then the hypothesis $h_j$ returned by $GenericLearner$ is s.t.\ $\error_{\PPP_j}(h_j)\leq\alpha+\Delta$.

For every $(y_1,\ldots,y_{|B|})\in \Pi_C(B)$, algorithm $GenericLearner$ adds to $H$ a hypothesis $f$ s.t.\ $\forall 1\leq \ell \leq |B|,\;f(b_\ell)=y_\ell$.
In particular, $H$ contains a hypothesis $h_j^*$ s.t.\ $h_j^*(x)=c_j^*(x)$ for every $x\in B$, that is, a hypothesis $h_j^*$ s.t.\ $\error_{\widetilde{D}}(h_j^*,c_j^*)=0$.
As event $E_1$ has occurred we have that this $h_j^*$ satisfies $\error_D(h_j^*,c_j^*)\leq \frac{2\alpha}{5}$.
Using the triangle inequality (and event $E_2$) we get that this $h_j^*$ satisfies $\error_{S|_j}(h_j^*)\leq \error_D(h_j^*,c_j^*) + \error_{S|_j}(c_j^*) \leq \frac{3\alpha}{5}+\Delta$.
Thus, event $E_3$ ensures that algorithm $GenericLearner$ chooses (using the exponential mechanism) a hypothesis $h_j\in H$ s.t.\ $\error_{S|_j}(h_j)\leq\frac{4\alpha}{5}+\Delta$. Event $E_2$ ensures, therefore, that this $h_j$ satisfies $\error_{\PPP_j}(h_j)\leq\alpha+\Delta$.
We will now show $E_1 \cap E_2 \cap E_3$ happens with high probability.

Standard arguments in learning theory state that (w.h.p.) the empirical error on a (large enough) random sample is close to the generalization error.
Specifically, by setting $n\geq O(\frac{1}{\alpha^2}\VC(C)\log(\frac{k}{\alpha\beta}))$, Theorem~\ref{thm:AgnosticGeneralization} ensures that Event $E_2$ occurs 
with probability at least $(1-\frac{2}{5}\beta)$.

Assuming that $n\geq m$ (the sample complexity of the sanitizer used in Step~5), with probability at least $(1-\frac{\beta}{5})$ for every $(h\oplus f)\in C^{\oplus}$ (i.e., for every $h,f\in C$) it holds that
\begin{eqnarray*}
\frac{\alpha}{5}&\geq& | Q_{(h{\oplus} f)}(D) - Q_{(h{\oplus} f)}(\widetilde{D}) |  \\
&=& \left| \frac{|\{x\in D : (h{\oplus} f)(x){=}1 \}|}{|D|} - \frac{|\{x\in \widetilde{D} : (h{\oplus} f)(x){=}1 \}|}{|\widetilde{D}|}  \right|  \\
&=& \left| \frac{|\{x\in D : h(x){\neq} f(x) \}|}{|D|} - \frac{|\{x\in \widetilde{D} : h(x){\neq} f(x) \}|}{|\widetilde{D}|}  \right|  \\
&=& \left| \error_D(h,f) - \error_{\widetilde{D}}(h,f)  \right|. 
\end{eqnarray*}
Event $E_1$ occurs therefore with probability at least $(1-\frac{\beta}{5})$.

The exponential mechanism ensures that the probability of event $E_3$ is at least $1-k|H| \cdot \exp(-\epsilon' \alpha m /10)$ (see Proposition \ref{prop:exp_mech}).
Note that $\log|H|\leq|B|\leq|\widetilde{D}| = O\left( \frac{\VC(C)}{\alpha^2}\log(\frac{1}{\alpha}) \right)$. Therefore, for
$n \geq O\left(\frac{\VC(C)}{\alpha^3\epsilon'}\log(\frac{1}{\alpha}) +\frac{1}{\alpha\epsilon'}\log(\frac{k}{\beta})\right)$,
Event $E_3$ occurs with probability at least $(1-\frac{\beta}{5})$. 

All in all, setting
$n \geq O\left(m+\frac{\VC(C)}{\alpha^3\epsilon'}\log(\frac{1}{\alpha}) +\frac{1}{\alpha\epsilon'}\log(\frac{k}{\beta}) + \frac{1}{\alpha^2}\VC(C)\log(\frac{k}{\alpha\beta})\right)$,
ensures that the probability of $GenericLearner$ failing is at most $\beta$.
\end{proof}

Theorem~\ref{thm:GeneralUpper} now follows by combining Lemma~\ref{lem:usingSan} and Corollary~\ref{cor:straightforward}.

\medskip

For certain concept classes, there are sanitizers with substantially lower sample complexity than the generic sanitizers. Combining Lemma \ref{lem:usingSan} with Proposition \ref{prop:sanPoint}, we obtain:

\begin{corollary}\label{cor:pointsAgnosticUpper}
There is an $(\alpha,\beta)$-PAC agnostic $k$-learner for $\point_X$ with sample complexity
$$O\left(\frac{\log(1/\alpha\beta\delta)}{\alpha\epsilon}+\frac{\log(1/\alpha)}{\alpha^3\epsilon'} +\frac{\log(k/\beta)}{\alpha\epsilon'} + \frac{\log(k/\alpha\beta)}{\alpha^2}\right).$$
Moreover, it is both $(\epsilon+k\epsilon',\delta)$ and $(\epsilon+\sqrt{2k\ln(1/\delta)}\epsilon'+2k\epsilon'^2,2\delta)$-differentially private.
\end{corollary}

Similarly, combining Lemma \ref{lem:usingSan} with Proposition \ref{prop:sanThreshold}, we obtain:

\begin{corollary}\label{cor:thresholdsAgnosticUpper}
There is an $(\alpha,\beta)$-PAC agnostic $k$-learner for $\thresh_X$ with sample complexity
$$O\left(\frac{2^{\log^*|X|} \cdot \log^*|X| \cdot \log \left(\frac{\log^*|X|}{\eps\delta}\right) \cdot \log(1/\beta) \cdot \log^{2.5}(1/\alpha)}{\alpha \eps} +\frac{\log(1/\alpha)}{\alpha^3\epsilon'} +\frac{\log(k/\beta)}{\alpha\epsilon'} + \frac{\log(k/\alpha\beta)}{\alpha^2}\right).$$
Moreover, it is both $(\epsilon+k\epsilon',\delta)$ and $(\epsilon+\sqrt{2k\ln(1/\delta)}\epsilon'+2k\epsilon'^2,2\delta)$-differentially private.
\end{corollary}

\subsection{Upper Bounds for Approximate Private Multi-Learners}

In this section we give two examples of cases where the sample complexity of private $k$-learning is of the same order as that of non-private $k$-learning (the sample complexity does not depend on $k$). Our algorithms are $(\eps, \delta)$-differentially private, and rely on stability arguments: the identity of the best $k$ concepts, as an entire vector, is unlikely to change on nearby $k$-labeled databases. Hence, it can be released privately.

The main technical tool we use is the $\adist$ algorithm of Smith and Thakurta \cite{ST13}. Our discussion follows the treatment of \cite{BNS13b}.

Recall that a quality function  $q:X^*\times \cF \rightarrow \N$ defines an {\em optimization problem} over the domain $X$ and a finite solution set $\cF$:  Given a database $S \in X^*$, find $f\in\cF$ that (approximately) maximizes $q(S,f)$. The sensitivity of a quality function, $\Delta q$, is the maximum over all $f\in\cF$ of the sensitivity of the function $q(\cdot, f)$. The algorithm $\adist$ privately identifies the \emph{exact} maximizer as long as it is sufficiently stable.

\begin{algorithm}
\caption{$\adist$}\label{alg:adist}
{\bf Input:} Privacy parameters $\epsilon,\delta$, database $S \in X^*$, sensitivity-1 quality function $q$
\begin{enumerate}[rightmargin=10pt,itemsep=1pt]

\item Let $f_1, f_2 \in \cF$ be the highest scoring and second-highest scoring solutions to $q(S, \cdot)$, respectively.

\item Let $\gap = q(S, f_1) - q(S, f_2)$, and $\widehat{\gap} = \gap + \Lap(1/\eps)$.

\item If $\widehat{\gap} < \frac{1}{\eps} \log \frac{1}{\delta}$, output $\bot$. Otherwise, output $f_1$.
\end{enumerate}
\end{algorithm}

\begin{proposition}[Properties of $\adist$ \cite{ST13}] \label{prop:adist}
\ \begin{enumerate}
\item Algorithm $\adist$ is $(\eps, \delta)$-differentially private.
\item When run on a database $S$ with $\gap > \frac{1}{\eps} \log \frac{1}{\delta\beta}$, Algorithm $\adist$ outputs the highest scoring solution $f_1$ with probability at least $1 - \beta$.
\end{enumerate}
\end{proposition}

\subsubsection{Learning Parities under the Uniform Distribution}

\begin{theorem}\label{thm:paritiesUniformUpper}
For every $k,d$ there exists an $(\alpha{=}0,\beta,\epsilon,\delta)$-PAC (non-agnostic) $k$-learner for $\parity_d$ under the uniform distribution with sample complexity $O(d\log(\frac{1}{\beta})+ \frac{1}{\epsilon}\log(\frac{1}{\beta\delta}))$.
\end{theorem}

Recall that (even without privacy constraints) the sample complexity of PAC learning $\parity_d$ under the uniform distribution is $\Omega(d)$. Hence the sample complexity of privately $k$-learning $\parity_d$ (non-agnostically) under the uniform distribution is of the same order as that of non-private $k$-learning.

For the intuition behind Theorem~\ref{thm:paritiesUniformUpper}, let $c_1,\dots,c_k$ denote the $k$ target concepts, and
consider the quality function $q(D,(h_1,\dots,h_k))=\max_{1\leq j\leq k}\{ \error_D(h_j,c_j) \}$. On a large enough sample $D$ we expect that $q(D,(h_1,\dots,h_k))\approx\frac{1}{2}$ for every $(h_1,\dots,h_k)\neq(c_1,\dots,c_k)$, while $q(D,(c_1,\dots,c_k))=0$. The $k$ target concepts can hence be privately identified (exactly) using stability techniques.

In order to make our algorithm computationally efficient, we apply the ``subsample and aggregate'' idea of Nissim et al.~\cite{NRS07}. We divide the input sample into a small number of subsamples, use Gaussian elimination to (non-privately) identify a candidate hypothesis vector on each subsample, and then select from these candidates privately.

\begin{algorithm}[h!]
\caption{$ParityLearner$}\label{alg:parityLearner}
{\bf Input:} Parameters $\epsilon,\delta$, and a $k$-labeled database $S$ of size $n=O(\frac{d}{\epsilon}\log(\frac{1}{\beta\delta}))$.\\
{\bf Output:} Hypotheses $h_1,\dots,h_k$.
\begin{enumerate}[rightmargin=10pt,itemsep=1pt]

\item Split $S$ into $m=O(\frac{1}{\epsilon}\log(\frac{1}{\beta\delta}))$ disjoint samples $S_1,\ldots,S_m$ of size $O(d)$ each. Initiate $Y$ as the empty multiset.

\item For every $1\leq t\leq m$:
\vspace{-5px}
\begin{enumerate}
	\item For every $1\leq j\leq k$ try to use Gaussian elimination to identify a parity function $y_j$ that agrees with the labels of the $j^{\text{th}}$ column of $S_t$.
  \item If a parity is identified for every $j$, then set $Y=Y\cup\{(y_1,...,y_k)\}$. Otherwise set $Y=Y\cup\{\bot\}$.
\end{enumerate}

\item Use algorithm $\AAA_{\text{dist}}$ with privacy parameters $\epsilon,\delta$ to choose and return a vector of $k$ parity functions $(h_1,\ldots,h_k)\in(\parity_d)^k$ with a large number of appearances in $Y$.

\end{enumerate}
\end{algorithm}

\begin{proof}[Proof of Theorem~\ref{thm:paritiesUniformUpper}]
The proof is via the construction of $ParityLearner$ (algorithm~\ref{alg:parityLearner}). First note that changing a single input element in $S$ can change (at most) one element of $Y$. Hence, applying (the $(\epsilon,\delta)$-private) algorithm $\AAA_\text{dist}$ on $Y$ preserves privacy (applying $ParityLearner$ on neighboring inputs amounts to executing $\AAA_\text{dist}$ on neighboring inputs).

Now fix $k$ target concepts $c_1,\dots,c_k\in\parity_d$ and let $S$ be a random $k$-labeled database containing $n$ i.i.d.\ elements from the uniform distribution $U_d$ over $X=\{0,1\}^d$, each labeled by $c_1,\dots,c_k$. 
Observe that (for every $1\leq t\leq m$) we have that $S_t$ contains i.i.d.\ elements from $U_d$ labeled by $c_1,\dots,c_k$.
We use $D_t$ to denote the unlabeled portion of $S_t$.
Standard arguments in learning theory (cf. Theorem~\ref{thm:Generalization}) state that for $|S_t|\geq O(d)$,
\begin{align*}
\Pr\left[\exists h,f\in\parity_d \text{ s.t. } 
\error_{U_d}(h,f)\geq\frac{1}{4}
\quad \wedge \quad \error_{D_t}(h,f)\leq\frac{1}{40}
\right]\leq\frac{1}{8}.
\end{align*}
The above inequality holds, in particular, for every hypothesis $h\in\parity_d$ and every target concept $c_j$, and hence,
\begin{align*}
\Pr\left[\exists h\in\parity_d \text{ and } j \text{ s.t. } 
\error_{U_d}(h,c_j)\geq\frac{1}{4}
\quad \wedge \quad \error_{D_t}(h,c_j)\leq\frac{1}{40}
\right]\leq\frac{1}{8}.
\end{align*}
Recall that under the uniform distribution, the only $h\in\parity_d$ s.t.\ $\error_{U_d}(h,c_j)\neq\frac{1}{2}$ is $c_j$ itself, and hence
$$
\Pr\left[\exists h\in\parity_d \text{ and } j \text{ s.t. } h\neq c_j \quad \wedge \quad \error_{D_t}(h,c_j)\leq\frac{1}{40}\right]\leq\frac{1}{8}.
$$
So, for every $1\leq t\leq m$, with probability $7/8$ we have that for every label column $j$ the only hypothesis with empirical error less than $\frac{1}{40}$ on $S_t$ is the $j^\text{th}$ target concept itself (with empirical error 0). 
In such a case, step~2a (Gaussian elimination) identifies exactly the vector of $k$ target concepts $(c_1,\dots,c_k)$.
Since $m\geq O(\log(\frac{1}{\beta}))$, the Chernoff bound ensures that except with probability $\beta/2$, the vector $(c_1,\dots,c_k)$ is identified in at least $3/4$ of the iterations of step~2.
Assuming that this is the case, the vector $(c_1,\dots,c_k)$ appears in $Y$ at least $3m/4$ times, while every other vector can appear at most $m/4$ times.
Provided that $m\geq O(\frac{1}{\epsilon}\log(\frac{1}{\beta\delta}))$, algorithm $\AAA_\text{dist}$ ensures that the $k$ target concepts are chosen with probability $1-\beta/2$.

All in all, algorithm $ParityLearner$ identifies the $k$ target concepts (exactly) with probability $1-\beta$, provided that $n\geq O(\frac{d}{\epsilon}\log(\frac{1}{\beta\delta}))$.
\end{proof}

\subsubsection{Learning Points}

We next show that the class of $\point_X$ can be (non-agnostically) $k$-learned using constant sample complexity, matching the non-private sample complexity.
\begin{theorem} \label{thm:pointsUpper}
For every domain $X$ and every $k\in\N$ there exists an $(\alpha,\beta,\epsilon,\delta)$-PAC (non-agnostic) $k$-learner for $\point_X$ with sample complexity $O(\frac{1}{\alpha\epsilon}\log(\frac{1}{\alpha\beta\delta}))$.
\end{theorem}

The proof is via the construction of Algorithm~\ref{alg:pointLearner}. The algorithm begins by privately identifying (using sanitization) a set of $O(1/\alpha)$ ``heavy'' elements in the input database, appearing $\Omega(\alpha)$ times. The $k$ labels of such a heavy element can be privately identified using stability arguments (since their duplicity in the database is large). The labels of a ``non-heavy'' element can be set to 0 since a target concept can evaluate to 1 on at most one such non-heavy element, in which case the error is small.

\paragraph{Notation.} We use $\#_S(x)$ to denote the duplicity of a domain element $x$ in a database $S$. For a distribution $\mu$ we denote $\mu(x)=\Pr_{\hat{x}\sim\mu}[\hat{x}=x]$.

\begin{algorithm}[h]
\caption{$PointLearner$}\label{alg:pointLearner}
{\bf Input:} Privacy parameters $\epsilon,\delta$, and a $k$-labeled database $S=(x_i,y_{i,1},\dots,y_{i,k})_{i=1}^n$. We use $D=(x_i)_{i=1}^n$ to denote the unlabeled portion of $S$.\\
{\bf Output:} Hypotheses $h_1,\dots,h_k$.
\begin{enumerate}[rightmargin=10pt,itemsep=1pt]

\item Let $\hat{D}\in X^m$ be an $(\frac{\epsilon}{2},\frac{\delta}{2})$-private $(\frac{\alpha}{30},\frac{\beta}{4})$-accurate sanitization of $D$ w.r.t.\ $\point_X$ (e.g., using Proposition~\ref{prop:sanPoint}).

\item Let $G=\{ x\in X : \frac{1}{m}\#_{\hat{D}}(x)\geq\alpha/15 \}$ be the set of all ``$\frac{\alpha}{15}$-heavy'' domain elements w.r.t.\ the sanitization $\hat{D}$. Note that $|G|\leq15/\alpha$.

\item Let $q$ be the quality function that on input a $k$-labeled database $S$, a domain element $x$, and a binary vector $\vec{v}\in\{0,1\}^k$, returns the number of appearances of $(x,\vec{v})$ in $S$. That is, $q(S,x,(v_1,\dots,v_k))=|\{ i : x_i=x \wedge y_{i,1}=v_1 \wedge \dots \wedge y_{i,k}=v_k \}|$.

\item Use algorithm $\AAA_{\text{dist}}$ with privacy parameters $\frac{\epsilon}{2},\frac{\delta}{2}$ to choose a set of vectors $V=\{\vec{v}_x\in\{0,1\}^k : x\in G\}$ maximizing $Q(S,V)=\min_{\vec{v}_x\in V}\{q(S,x,\vec{v}_x)\}$. That is, we use algorithm $\AAA_\text{dist}$ to choose a set of $|G|$ vectors -- a vector $\vec{v}_x$ for every $x\in G$ -- such that the minimal number of appearances of an entry $(x,\vec{v}_x)$ in the database $S$ is maximized.

\item For $1\leq j \leq k$: If the $j^\text{th}$ entry of every $\vec{v}_x\in V$ is 0, then set $h_j\equiv0$. Otherwise, let $x$ be s.t. $\vec{v}_x\in V$ has 1 as its $j^\text{th}$ entry, and define $h_j:X\rightarrow\{0,1\}$ as $h_j(y)=1$ iff $y=x$. 

\item Return $h_1,\dots,h_k$.

\end{enumerate}
\end{algorithm}

\begin{proof}
The proof is via the construction of $PointLearner$ (algorithm~\ref{alg:pointLearner}).
First note the algorithm only access the input database using sanitization on step~1, and using algorithm $\AAA_\text{dist}$ on step~4. By composition theorem~\ref{thm:composition}, algorithm $PointLearner$ is $(\epsilon,\delta)$-differentially private.

Let $\mu$ be a distribution over $X$, and let $c_1,\dots,c_k\in\point_X$ be the fixed target concepts.
Consider the execution of $PointLearner$ on a database $S=(x_i,y_{i,1},\dots,y_{i,k})_{i=1}^n$ sampled from $\mu$ and labeled by $c_1,\dots,c_k$. We use $D$ to denote the unlabeled portion of $S$, $\hat{D}$ for the sanitization of $D$ constructed on step~1, and write $m=|\hat{D}|$. Define the following good events.

\begin{enumerate}[label=$E_{\arabic*}:$]

\item For every $x\in X$ s.t.\ $\mu(x)\geq\alpha$ it holds that $\frac{1}{n}\#_S(x)\geq\alpha/10$.

\item For every $x\in X$ we have that $|\frac{1}{m}\#_{\hat{D}}(x) - \frac{1}{n}\#_S(x)| \leq \alpha/30$.

\item Algorithm $\AAA_\text{dist}$ returns a vector set $V$ s.t.\ $q(S,x,\vec{v}_x)\geq1$ for every $x\in G$.

\end{enumerate}

We now argue that when these three events happen algorithm $PointLearner$ returns good hypotheses.
%First, observe that the algorithm does not fail on step~2: 
%The algorithm fails if $|G|>30/\alpha$, that is if there are more than $30/\alpha$ elements $x$ with $san(x)\geq\alpha/15$. As event $E_2$ has occurred, this means that there are more than $30/\alpha$ elements $x$ that appear at least $\alpha n/30$ times in $S$, contradicting the fact that the size of $S$ is $n$.
First, observe that the set $G$ contains every element $x$ s.t. $\mu(x)\geq\alpha$: Let $x$ be s.t.\ $\mu(x)\geq\alpha$. As event $E_1$ has occurred, we have that $\frac{1}{n}\#_S(x)\geq\alpha/10$. As event $E_2$ has occurred, we have that $\frac{1}{m}\#_{\hat{D}}(x)\geq\alpha/15$, and therefore $x\in G$.

Note that if $q(S,x,\vec{v})\geq1$ then the example $x$ is labeled as $\vec{v}$ by the target concepts. Thus, as event $E_3$ has occurred, for every $\vec{v}_x\in V$ it holds that $\vec{v}_x=(c_1(x),\dots,c_k(x))$.
Now let $h_j$ be the $j^{\text{th}}$ returned hypothesis. We next show that $h_j$ is $\alpha$-good.
If $h\not\equiv0$, then let $x$ be the unique element s.t.\ $h_j(x)=1$, and note that (according to step~5) the $j^\text{th}$ entry of $\vec{v}_x$ is 1, and hence, $c_j(x)=1$. So $h_j=c_j$ (since $c_j$ is a concept in $\point_X$).

If $h_j\equiv0$ then the $j^\text{th}$ entry of every $\vec{v}_x\in V$ is 0. Note that in such a case $h_j$ only errs on the unique element $x$ s.t.\ $c_j(x)=1$, and it suffices to show that $\mu(x)<\alpha$. Assume towards contradiction that $\mu(x)\geq\alpha$. As before, event $E_1\cap E_2$ implies that $x\in G$. As event $E_3$ has occurred, we also have that $\vec{v}_x\in V$ is s.t.\ $q(S,x,\vec{v}_x)\geq1$, and the example $x$ is labeled as $\vec{v}_x$ by the target concepts. This contradicts the assumption that the $j^\text{th}$ entry of $\vec{v}_x\in V$ is 0.

Thus, whenever $E_1 \cap E_2 \cap E_3$ happens, algorithm $PointLearner$ returns $\alpha$-good hypotheses. We will now show $E_1 \cap E_2 \cap E_3$ happens with high probability.
Provided $n\geq O(\frac{1}{\alpha\epsilon}\log(\frac{1}{\alpha\delta}))$, event $E_2$ is guaranteed to hold with all but probability $\beta/4$ by the utility properties of the sanitizer used on step~1. See Proposition~\ref{prop:sanPoint}.

Theorem~\ref{thm:Generalization} (VC bound) ensures that event $E_1$ holds with probability $1-\beta/4$, provided that $n\geq O(\frac{1}{\alpha}\log(\frac{1}{\alpha\beta}))$.
To see this, let $z\equiv0$ denote the constant 0 hypothesis, and consider the class $C=\point_X\cup\{z\}$. Note that $\VC(C)=1$. Hence, Theorem~\ref{thm:generalization} states that, with all but probability $1-\beta/4$, for every $c\in\point_x$ s.t.\ $\error_{\mu}(c,z)\geq\alpha$ it holds that $\error_D(c,z)\geq\alpha/10$. That is, with all but probability $1-\beta/4$, for every $x\in X$ s.t.\ $\mu(x)\geq\alpha$ it holds that $\frac{1}{n}\#_D(x)=\frac{1}{n}\#_S(x)\geq\alpha/10$.

Before analyzing event $E_3$, we show that if $E_2$ occurs, then every $x\in G$ is s.t. $\#_S(x)\geq\alpha/30$. Let $x\in G$, that is, $x$ s.t.\ $\frac{1}{m}\#_{\hat{D}}(x)\geq\alpha/15$. Assuming event $E_2$ has occurred, we therefore have that $\frac{1}{n}\#_S(x)\geq\alpha/30$. So every $x\in G$ appears in $S$ at least $\alpha n/30$ times with the labels $(c_1(x),\dots,c_k(x)) \triangleq \vec{c}(x)$. Thus, $q(S,x,\vec{c}(x))\geq\alpha n/30$. In addition, for every $\vec{v}\neq\vec{c}(x)$ it holds that $q(S,x,\vec{v})=0$, since {\em every} appearance of the example $x$ is labeled by the target concepts. Hence, provided that $n\geq O(\frac{1}{\alpha\epsilon}\log(\frac{1}{\beta\delta}))$, algorithm $\AAA_{\text{dist}}$ ensures that event $E_3$ happens with probability at least $1-\beta/2$.

Overall, $E_1 \cap E_2 \cap E_3$ happens with probability at least $1-\beta$.
\end{proof}

\section{Approximate Privacy Lower Bounds from Fingerprinting Codes}\label{sec:FPC}

In this section, we show how fingerprinting codes can be used to obtain $\poly(k)$ lower bounds against privately learning $k$ concepts, even for very simple concept classes. Fingerprinting codes were introduced by Boneh and Shaw \cite{BS98} to address the problem of watermarking digital content. The connection between fingerprinting codes and differential privacy lower bounds was established by Bun, Ullman, and Vadhan \cite{BUV14} in the context of private query release, and has since been extended to a number of other differentially private analyses \cite{BST14, DTTZ14, SU15, BNSV15}.

A (fully-collusion-resistant) fingerprinting code is a scheme for distributing codewords $\codeword_1, \dots, \codeword_n$ to $n$ users that can be uniquely traced back to each user. Moreover, if any group of users combines its codewords into a pirate codeword $\pirateword$, then the pirate codeword can still be traced back to one of the users who contributed to it. Of course, without any assumption on how the pirates can produce their combined codeword, no secure tracing is possible. To this end, the pirates are constrained according to a \emph{marking assumption}, which asserts that the combined codeword must agree with at least one of the pirates' codeword in each position. Namely, at an index $j$ where $\codeword_{ij} = b$ for every $i \in b$, the pirates are constrained to output $\pirateword$ with $\pirateword_j = b$ as well. 

To illustrate our technique, we start with an informal discussion of how the original Boneh-Shaw fingerprinting code yields an $\tilde{\Omega}(k^{1/3})$ sample complexity lower bound for multi-learning threshold functions. For parameters $n$ and $k$, the $(n, k)$-Boneh-Shaw codebook is a matrix $\codebook \in \bits^{n \times k}$, whose rows $\codeword_i$ are the codewords given to users $i = 1, \dots, n$. The codebook is built from a number of highly structured columns, where a ``column of type $i$'' consists of $n$ bits where the first $i$ bits are set to $1$ and the last $n-i$ bits are set to $0$. For $i = 1, \dots, n-1$, each column of type $i$ is repeated a total of $k/(n-1)$ times, and the codebook $\codebook$ is obtained as a random permutation of these $k$ columns. The security of the Boneh-Shaw code is a consequence of the secrecy of this random permutation. If a coalition of pirates is missing the codeword of user $i$, then it is unable to distinguish columns of type $i-1$ from columns of type $i$. Hence, if a pirate codeword is too consistent with a user $i$'s codeword in both the columns of type $i-1$ and the columns of type $i$, a tracing algorithm can reasonably conclude that user $i$ contributed to it. Boneh and Shaw showed that such a code is indeed secure for $k = \tilde{O}(n^3)$.

To see how this fingerprinting code gives a lower bound for multi-learning thresholds, consider thresholds over the data universe $X = \{1, \dots, |X|\}$ for $|X| \ge n$. The key observation is that each column of the Boneh-Shaw codebook can be obtained as a labeling of the examples $1, \dots, n$ by a threshold concept. Namely, a column of type $i$ is the labeling of $1, \dots, n$ by the concept $c_i$. Now suppose a coalition of users $T \subseteq [n]$ constructs a database $S$ where each row is an example $i \in T$ together with the labels $\codeword_{i1}, \dots, \codeword_{ik}$ coming from the codeword given to user $i$. Let $(h_1, \dots, h_k)$ be the hypotheses produced by running a threshold multi-learner on the database. If every user has a bit $b$ at index $j$ of her codeword, then the hypothesis produced by the learner must also evaluate to $b$ on most of the examples. Thus, the empirical averages of the hypotheses $(h_1, \dots, h_k)$ on the examples can be used to obtain a pirate codeword satisfying the marking assumption. The security of the fingerprinting code, i.e. the fact that this codeword can be traced back to a user $i \in T$, implies that the learner cannot be differentially private. Hence, $n$ samples is insufficient for privately learning $k = \tilde{O}(n^3)$ threshold concepts, giving a sample complexity lower bound of $\tilde{\Omega}(k^{1/3})$.

The lower bounds in this section are stated for empirical learning, but extend to PAC learning by Theorem \ref{thm:pac-to-empirical}. We also remark that they hold against the relaxed privacy notion of \emph{label privacy}, where differential privacy only needs to hold with respect to changing the labels of one example.

\subsection{Fingerprinting Codes}

An $(n, k)$\emph{-fingerprinting code} consists of a pair of randomized algorithms $(\gen, \trace)$. The parameter $n$ is the number of users supported by the fingerprinting code, and $k$ is the length of the code. The codebook generator $\gen$ produces a \emph{codebook} $\codebook \in \bits^{n \times k}$. Each row $\codeword_i \in \bits^k$ of $\codebook$ is the \emph{codeword} of user $i$. For a subset $T \subseteq [n]$, we let $\codebook_T$ denote the set $\{\codeword_i : i \in T\}$ of codewords belonging to users in $T$. The accusation algorithm $\trace$ takes as input a pirate codeword $\pirateword$ and accuses some $i \in [n]$ (or $\bot$ if it fails to accuse any user). 

We define the feasible set of pirate codewords for a coalition $T$ and codebook $\codebook$ by
\[F(\codebook_T) = \{\pirateword \in \{0, 1\}^k : \forall j = 1, \dots, k \  \exists i \in S \text{ s.t. } \codeword_{ij} = \pirateword_j\}.\]
The basic marking assumption is that the pirate codeword $\pirateword \in F(\codebook_T)$. We say column $j$ is \emph{$b$-marked} if $\codeword_{ij} = b$ for every $i \in [n]$.

\begin{definition}[Fingerprinting Codes] \label{def:fpc}
For $n, k \in \N$ and $\sec \in (0,1]$, a pair of algorithms $(\gen, \trace)$ is an \emph{$(n, k)$-fingerprinting code with security $\sec$} if $\gen$ outputs a codebook $\codebook \in \bits^{n \times k}$ and for every (possibly randomized) adversary $\fpadv$, and every coalition $T \subseteq [n]$, if we take $\pirateword \getsr \fpadv(\codebook_T)$, then the following properties hold.
\begin{description}
\item[Completeness:]
$
\prob{\pirateword \in F(\codebook_T) \land \trace(\pirateword) = \bot} \leq \sec,
$
\item[Soundness:]
$
\prob{\trace(\pirateword) \in [n] \setminus T} \leq \sec,
$
\end{description}
Each probability is taken over the coins of $\gen, \trace$, and $\fpadv$.  The algorithms $\gen$ and $\trace$ may share a common state, which is hidden to ease notation.
\end{definition}

\subsection{Lower Bound for Improper PAC Learning}

Our lower bounds for multi-learning follow from constructions of fingerprinting codes with additional structural properties.

\begin{definition} \mnote{This is the strongest definition of compatibility -- can probably be relaxed.}
Let $C$ be a concept class over a domain $X$. An $(n, k)$-fingerprinting code $(\gen, \trace)$ is \emph{compatible with concept class $C$} if there exist $x_1, \dots, x_n \in X$ such that for every codebook $\codebook$ in the support of $\gen$, there exist concepts $c_1, \dots, c_k$ such that $\codeword_{ij} = c_j(x_i)$ for every $i = 1, \dots, n$ and $j = 1, \dots, k$.
\end{definition}

\begin{theorem} \label{thm:fpclb-pac}
Suppose there exists an $(n, k)$-fingerprinting code compatible with a concept class $C$ with security $\sec$. Let $\alpha \le 1/3$, $\beta, \eps > 0$, and $\delta < \frac{1-\sec-\beta}{n}-e^\eps\sec$. Then every (improper) $(\alpha, \beta)$-accurate and $(\eps, \delta)$-differentially private empirical $k$-learner for $C$ requires sample complexity greater than $n$.
\end{theorem}

The proof of Theorem~\ref{thm:fpclb-pac} follows the ideas sketched above.

\begin{proof}
Let $(\gen, \trace)$ be an $(n, k)$-fingerprinting code compatible with the concept class $C$, and let $x_1, \dots, x_n \in X$ be its associated universe elements. Let $D = (x_1, \dots, x_n)$ and let $\cA$ be an $(\alpha, \beta)$-accurate empirical $k$-learner for $C$ with sample complexity $n$. We will use $\cA$ to design an adversary $\fpadv$ against the fingerprinting code.

Let $T \subseteq [n]$ be a coalition of users, and consider a codebook $\codebook \getsr \gen$. The adversary strategy $\fpadv(\codebook_T)$ begins by constructing a labeled database $S = (S_i)_{i = 1}^n$ by setting $S_i = (x_i, \codeword_{i1}, \dots, \codeword_{ik})$ for each $i \in T$ and to a nonce row for $i \notin T$. It then runs $\cA(S)$ obtaining hypotheses $(h_1, \dots, h_k)$. Finally, it computes for each $j = 1, \dots, k$ the averages
\[h_j(D) = \frac{1}{n} \sum_{i = 1}^n h_j(x_i)\]
and produces a pirate word $\pirateword$ by setting each $\pirateword_j$ to the value of $a_j$ rounded to $0$ or $1$.

Now consider the coalition $T = [n]$.  Since the fingerprinting code is compatible with $C$, each column $(\codeword_{1j}, \dots, \codeword_{nj}) = (c_j(x_1), \dots, c_j(x_n))$ for some concept $c_j \in C$. Thus, if the hypotheses $(h_1, \dots, h_k)$ are $\alpha$-accurate for $(c_1, \dots, c_k)$ on $S$, then $\pirateword \in F(\codebook_T) = F(\codebook)$. Therefore, by the completeness property of the code and the $(\alpha, \beta)$-accuracy of $\cA$, we have 
\[\prob{\trace(\fpadv(\codebook)) \ne \bot} \ge 1 - \sec - \beta.\]
In particular, there exists an $i^*$ for which
\[\prob{\trace(\fpadv(\codebook)) = i^*} \ge \frac{1 - \sec - \beta}{n}.\]
On the other hand, by the soundness property of the code,
\[\prob{\trace(\fpadv(\codebook_{-i^*})) = i^*} \le \sec.\]
Thus, $\cA$ cannot be $(\eps, \delta)$-differentially private whenever
\[\frac{1 - \sec - \beta}{n} > e^\eps \cdot \sec + \delta.\]
\end{proof}

\begin{remark} \label{rem:fpclb-pac-alpha}
If we additionally assume that there exists an element $x_0 \in X$ with $c_1(x_0) = c_2(x_0) = \dots = c_k(x_0)$, then we can use a ``padding'' argument to obtain a stronger lower bound of $n / 3\alpha$. More specifically, suppose $c_1(x_0) = \dots = c_k(x_0) = 0$. We pad the database $S$ constructed above with $(1/3\alpha - 1)n$ copies of the junk row $(x_0, 0, \dots, 0)$. Now if a hypothesis $h$ is $\alpha$-accurate for a $0$-marked column, it's empirical average will be at most $\alpha$. On the other hand, an $\alpha$-accurate hypothesis for a $1$-marked column will have empirical average at least $2\alpha$. Since there is a gap between these two quantities, a pirate algorithm can still turn an accurate vector of $k$ hypotheses into a feasible codeword.
\end{remark}

As observed earlier, the $(n, k)$-Boneh-Shaw code is compatible with the concept class $\thresh_X$ for any $|X| \ge n$. Thus, instantiating Theorem \ref{thm:fpclb-pac} (and Remark \ref{rem:fpclb-pac-alpha}) with the Boneh-Shaw code yields a lower bound for $k$-learning thresholds.

\begin{lemma}[\cite{BS98}]
Let $X$ be a totally ordered domain with $|X| \ge n$ for some $n \in \N$. Then there exists an $(n, k)$-fingerprinting code compatible with the concept class $\thresh_X$ with security $\sec$ as long as $k \ge 2n^3\log(2n/\sec)$.
\end{lemma}

\begin{corollary} \label{cor:thresholds}
Every improper $(\alpha, \beta)$-accurate and $(\eps = O(1), \delta = o(1/n))$-differentially private empirical $k$-learner for $\thresh_X$ requires sample complexity $\min\{|X|, \tilde{\Omega}(k^{1/3} / \alpha)\}$.
\end{corollary}

\paragraph{Discussion.} Compatibility with a concept class is an interesting measure of the complexity of a fingerprinting code which warrants further attention. Peikert, shelat, and Smith \cite{PSS03} showed that structural constraints (related to compatibility) on a fingerprinting code give a lower bound on its length beyond the general lower bound of $k = \tilde{\Omega}(n^2)$ for arbitrary fingerprinting codes. In particular, they showed that the length $k = \tilde{O}(n^3)$ of the Boneh-Shaw code is essentially tight for the ``multiplicity paradigm'', where a codebook is a random permutation of a fixed set of columns, each repeated the same number of times. We take this as evidence that our $\tilde{\Omega}(k^{1/3})$ lower bound for $\thresh_X$ cannot be improved via compatible fingerprinting codes. However, closing the gap between our lower bound and the upper bound of roughly $\sqrt{k}$ remains an intriguing open question.

A natural avenue for obtaining stronger $\poly(k)$ lower bounds for private $k$-learning is to identify compatible fingerprinting codes with shorter length. Tardos \cite{Tardos08} showed the existence of an $(n, k)$-fingerprinting code of optimal length $k = \tilde{O}(n^2)$ (see Proposition \ref{prop:tardos}). The construction of his code differs significantly from multiplicity paradigm: for each column $j$ of the Tardos code, a bias $p_j \in (0, 1)$ is sampled from a fixed distribution, and then each bit of the column is sampled i.i.d. with bias $p_j$. Hence, the columns of the Tardos code are supported on all bit vectors in $\bits^n$. This means that for a concept class $C$ to be compatible with the $(n, k)$-Tardos code, it must be the case that $\VC(C) \ge n$. Thus, the lower bound one obtains against $k$-learning $C$ only matches the lower bound for PAC learning $C$ (without privacy). It would be very interesting to construct a fingerprinting code of optimal length $k = \tilde{O}(n^2)$ with substantially fewer than $2^n$ column types (and hence compatible with a concept class of VC-dimension smaller than $n$).

\subsection{Lower Bound for Agnostic Learning}

In the agnostic learning model, a learner has to perform well even when the columns of a multi-labeled database do not correspond to any concept. This allows us to apply the argument of Theorem \ref{thm:fpclb-pac} without the constraint of compatibility. The result is that \emph{any} fingerprinting code, in particular one with optimal length, gives an agnostic learning lower bound for any non-trivial concept class.

\begin{theorem} \label{thm:fpclb-agnostic}
Suppose there exists an $(n, k)$-fingerprinting code with security $\sec$. Let $C$ be a concept class with at least two distinct concepts. Let $\alpha \le 1/3$, $\beta, \eps > 0$, and $\delta < \frac{1-\sec-\beta}{n}-e^\eps\sec$. Then every (improper) agnostic $(\alpha, \beta)$-accurate and $(\eps, \delta)$-differentially private empirical $k$-learner for $C$ requires sample complexity greater than $n$.
\end{theorem}

\begin{proof}
The proof follows in much the same way as that of Theorem \ref{thm:fpclb-pac}. Let $(\gen, \trace)$ be an $(n, k)$-fingerprinting code, and let $x \in X$ be such that there exist $c_0, c_1 \in C$ with $c_0(x) = 0$ and $c_1(x) = 1$.  Let $\cA$ be an agnostic $(\alpha, \beta)$-accurate empirical $k$-learner for $C$ with sample complexity $n$. Define a the fingerprinting code adversary $\fpadv$ just as in Theorem \ref{thm:fpclb-pac}. Namely, $\fpadv$ constructs examples of the form $(x, w_{i1}, \dots, w_{ij})$ with the available rows of the fingerprinting code, runs $\cA$ on the result, and returns the rounded empirical averages of the $k$ resulting hypotheses.

To show that $\cA$ cannot be $(\eps, \delta)$-differentially private, it suffices to show that if $\cA$ produces accurate hypotheses $h_1, \dots, h_k$, then the pirate codeword produced by $\fpadv$ is feasible. To see this, suppose $h_1, \dots, h_k$ are accurate, i.e.
\[ \max_{1\leq j\leq k}\left(\error_{S|_j}(h_j)  -  \min_{c \in C}\left(\error_{S|_j}(c)\right) \right) \le \alpha.\]
Let column $j$ of the codebook $\codebook$ be $0$-marked, i.e. $\codeword_{ij} = 0$ for all $i \in [n]$. Recall that $c_0(x) = 0$, and hence $\error_{S|_j}(c_0)  = 0$. Therefore, since hypothesis $h_j$ is $\alpha$-accurate, we have $\error_{S|_j}(h_j) \le \alpha$. This implies that bit $\pirateword_j$ of the pirate codeword is $0$. An identical argument shows that the bits of the pirate codeword in the $1$-marked columns are also $1$. Thus, if $\cA$ produces accurate hypotheses, the pirate codeword produced by $\fpadv$ is feasible. The rest of the argument in the proof of Theorem \ref{thm:fpclb-pac} completes the proof.
\end{proof}

\begin{remark}
Just as in Remark \ref{rem:fpclb-pac-alpha}, a padding argument shows how to obtain a lower bound of $n / 3\alpha$ under some additional assumptions on $C$, e.g. if the distinct concepts also share a common point $x'$ with $c_0(x') = c_1(x')$.
\end{remark}

\begin{proposition}[\cite{Tardos08}] \label{prop:tardos}
For $n \in \N$ and $\sec \in (0, 1)$, there exists an $(n, k)$-fingerprinting code with security $\sec$ as long as $k = O(n^2 \log (n/\sec))$.
\end{proposition}

\begin{corollary}\label{cor:OmegaSqrtK}
Every improper agnostic $(\alpha, \beta)$-accurate and $(\eps = O(1), \delta = o(1/n))$-differentially private empirical $k$-learner for $\point_X, \thresh_X, \parity_d$ requires sample complexity $\min\{|X|, \tilde{\Omega}(k^{1/2})\}$.
\end{corollary}

The same proof yields a lower bound for agnostically learning parities under the uniform distribution.

\begin{proposition} \label{prop:fpclb-parities}
Suppose there exists an $(n, k)$-fingerprinting code with security $\sec$. Let $\alpha \le 1/6, \beta > 0$ and $d = \log n$. Then every (improper) agnostic $(\alpha, \beta, \eps = O(1), \delta = o(1/n))$-PAC $k$-learner for $\parity_d$ requires sample complexity $\Omega(n)$.
\end{proposition}

\begin{proof}[Proof sketch]
By Lemma \ref{thm:pac-to-empirical}, it is enough to rule out a private empirical learner for a database whose $n$ examples are the distinct binary strings in $\{0, 1\}^d$. To do so, we follow the proof of Theorem \ref{thm:fpclb-agnostic}, highlighting the changes that need to be made. First, we let $c_0$ be the all-zeroes concept, and let $c_1$ be an arbitrary other parity function. Second, $\fpadv$ instead constructs examples of the form $(x_i, \codeword_{i1}, \dots, \codeword_{ik})$ where $x_i$ is the $i$th binary string. Finally, when converting the hypotheses $(h_1, \dots, h_k)$ into a feasible codeword, we instead set $\pirateword_j$ to $0$ if $h_j(D) \le \alpha$, and set $\pirateword_j$ to $1$ if $h_j(D) \ge \frac{1}{2} - \alpha$. This works because, while $\error_{S|_j}(c_0) = 0$ with respect to $0$-marked columns, any concept (and in particular, $c_1$) has error $\frac{1}{2}$ with respect to $1$-marked columns.
\end{proof}

\section{Examples where Direct Sum is Optimal }\label{sec:pureLowerBounds}
In this section we show several examples for cases where the direct sum is (roughly) optimal.
As we saw in Section~\ref{sec:FPC}, with $(\epsilon,\delta)$-differential privacy, every non-trivial {\em agnostic} $k$-learner requires sample complexity $\Omega(\sqrt{k})$.
We can prove a similar result for $\epsilon$-private learners, that holds even for non-agnostic learners:
\begin{theorem}\label{thm:OmegaK}
Let $C$ be any non-trivial concept class over a domain $X$ (i.e.,\ $|C|\geq2$).
Every proper or improper $(\alpha,\beta{=}\frac{1}{2},\epsilon)$-private PAC $k$-learner for $C$ requires sample complexity $\Omega(k/\epsilon)$.
\end{theorem}

In~\cite{BKN10,BNS13a,BNS13b}, Beimel et al.\ presented an agnostic proper learner for $\point_X$ with sample complexity $O_{\alpha,\beta,\epsilon,\delta}(1)$ under $(\epsilon,\delta)$-privacy, and an agnostic improper learner for $\point_X$ with sample complexity $O_{\alpha,\beta,\epsilon,\delta}(1)$ under $\epsilon$-privacy.
Hence, using Observation~\ref{obs:directSum} (direct sum) with their results yields an $(\alpha,\beta,\epsilon,\delta)$-PAC agnostic proper $k$-learner for $\point_X$ with sample complexity $\tilde{O}_{\alpha,\beta,\epsilon,\delta}(\sqrt{k})$, and an $(\alpha,\beta,\epsilon)$-PAC agnostic improper $k$-learner for $\point_X$ with sample complexity $\tilde{O}_{\alpha,\beta,\epsilon}(k)$. As supported by our lower bounds (Corollary~\ref{cor:OmegaSqrtK} and Theorem~\ref{thm:OmegaK}), those learners have roughly optimal sample complexity (ignoring the dependency in $\alpha,\beta,\epsilon,\delta$ and logarithmic factors in $k$).

\begin{proof}[Proof of Theorem~\ref{thm:OmegaK}]
The proof is based on a packing argument~\cite{HT10,BKN10}.
Let $x\in X$ and $f,g\in C$ be s.t.\ $f(x)\neq g(x)$.
Let $\mu$ denote the constant distribution over $X$ giving probability 1 to the point $x$.
Note that $\error_{\mu}(f,g)=1$. Moreover, observe that for every concept $h$, if $\error_{\mu}(h,f)<1$ then $h(x)=f(x)$, and similarly with $h,g$.

Let $\cal A$ be an $(\alpha,\beta,\epsilon)$-private PAC $k$-learner for $C$ with sample complexity $n$.
For every choice of $k$ target functions $(c_1,\dots,c_k)=\vec{c}\in\{f,g\}^k$, let $S_{\vec{c}}$ denote the $k$-labeled database containing $n$ copies of the point $x$, each of which is labeled by $c_1,\dots,c_k$. 
Without loss of generality, we can assume that on such databases $\cal A$ returns hypotheses in $\{f,g\}$ (since under $\mu$ we can replace an arbitrarily chosen hypothesis $h$ with $f$ if $f(x)=h(x)$ or with $g$ if $g(x)=h(x)$).
Therefore, by the utility properties of $\cal A$, for every $\vec{c}=(c_1,\dots,c_k)\in\{f,g\}^k$ we have that $\Pr_{\cal A}[{\cal A}(S_{\vec{c}})=(c_1,\dots,c_k)]\geq\frac{1}{2}$. By changing the database $S_{\vec{c}}$ to $S_{\vec{c'}}$ one row at a time while applying the differential privacy constraint, we see that
$$\Pr_{\cal A}[{\cal A}(S_{\vec{c}})=(c'_1,\dots,c'_k)]\geq\frac{1}{2}e^{-\epsilon n}.$$
Since the above inequality holds for every two databases $S_{\vec{c}}$ and $S_{\vec{c'}}$, we get
$$\frac{1}{2}\geq\Pr_{\cal A}[{\cal A}(S_{\vec{c}})\neq(c_1,\dots,c_k)]\geq(2^k-1)\frac{1}{2}e^{-\epsilon n}.$$
Solving for $n$, this yields $n=\Omega(k/\epsilon)$.
\end{proof}

\begin{remark}
The above proof could easily be strengthened to show that $n=\Omega(\frac{k}{\alpha\epsilon})$, provided that $C$ contains two concepts $f,g$ s.t.\ $\exists x,y\in X$ for which $f(x)\neq g(x)$ and $f(y)=g(y)$. 
\end{remark}

The following lemma shows that the sample complexities of properly and improperly learning parities under the uniform distribution are the same. Thus, for showing lower bounds, it is without loss of generality to consider proper learners.
 
\begin{lemma} \label{lem:proper-parities}
Let $\alpha < 1/4$. Let $\cA$ be a (possibly improper) $(\alpha, \beta, \eps, \delta)$-PAC $k$-learner for $\parity_d$ under the uniform distribution with sample complexity $n$. Then there exists a proper $(\alpha' = 0, \beta, \eps, \delta)$-PAC $k$-learner $\cA'$ for $\parity_d$ (under the uniform distribution) with sample complexity $n$.
\end{lemma} 

\begin{proof}
The algorithm $\cA'$ runs $\cA$ and ``rounds'' each hypothesis $h_j$ produced to the nearest parity function. That is, it outputs $(h_1', \dots, h_k')$ where $h_j'$ is a parity function that minimizes $\Pr_{x \sim U_d} [h_j'(x) \ne h_j(x)]$. Since this is just post-processing of the differentially private algorithm $\cA$, the proper learner $\cA$ remains $(\eps, \delta)$-differentially private.

Now suppose $(h_1, \dots, h_k)$ is $\alpha$-accurate for parity functions $(c_1, \dots, c_k)$ on the uniform distribution. Then for each $j$,
\begin{align*}
\Pr_{x \sim U_d} [h_j'(x) \ne c_j(x)] &\le \Pr_{x \sim U_d} [h_j'(x) \ne h_j(x)] + \Pr_{x \sim U_d} [h_j(x) \ne c_j(x)] \\
&\le 2\Pr_{x \sim U_d} [h_j(x) \ne c_j(x)] \\
&\le 2\alpha.
\end{align*}
Hence, $\error_{U_d}(h_j', c_j) < 1/2$. Since the error of any parity function from $c_j$ (other than $c_j$ itself) is exactly $1/2$ under the uniform distribution, we conclude that $(h_1', \dots, h_k')$ is in fact $0$-accurate for $(c_1, \dots, c_k)$.
\end{proof}
 
\begin{theorem}\label{thm:paritiesLowerBound}
Let $\alpha < \frac{1}{4}$. Every $(\alpha,\beta{=}\frac{1}{2},\epsilon)$-PAC $k$-learner for $\parity_d$ (under the uniform distribution) requires sample complexity $\Omega(kd/\epsilon)$.
\end{theorem}

As we saw in Example~\ref{eg:PureParityUpper}, applying direct sum for $k$-learning parities results in a proper agnostic $(\alpha,\beta,\epsilon)$-PAC $k$-learner for $\parity_d$ with sample complexity $O_{\alpha,\beta,\epsilon}(kd+k\log k)$. As stated by Theorem~\ref{thm:paritiesLowerBound}, this is the best possible (ignoring logarithmic factors and the dependency in $\alpha,\beta,\epsilon$).

\begin{proof}[Proof of Theorem~\ref{thm:paritiesLowerBound}]
The proof is based on a packing argument~\cite{HT10,BKN10}.
Let $\cal A$ be an $(\alpha,\beta,\epsilon)$-PAC $k$-learner for $\parity_d$ with sample complexity $n$. By Lemma \ref{lem:proper-parities}, we may assume $\cA$ is proper and learns the hidden concepts exactly.
%Observe that for every $f\neq g\in\parity_d$ it holds that $\error_{U_d}(f,g)=\frac{1}{2}$.
%Hence, under $U_d$, a proper $k$-learner for $\parity_d$ must (w.h.p.) identify the $k$ target concepts exactly.
%As we only consider the uniform distribution $U_d$, we can assume w.l.o.g.\ that $\cal A$ is a proper learner since an arbitrarily chosen hypothesis $h$
 %can be replaced with ${\rm argmin}_{f\in\parity_d}\{\error_{\mu}(f,h)\}$ (only one such $f\in\parity_d$ can have error less than $1/2$).

For every choice of $k$ parity functions $(c_1,\dots,c_k)=\vec{c}\in(\parity_d)^k$, let $S_{\vec{c}}$ denote a random $k$-labeled database containing $n$ i.i.d.\ elements from $U_d$, each labeled by $(c_1,\dots,c_k)$. By the utility properties of $\cal A$ we have that $\Pr_{U_d,\AAA}[\AAA(S_{\vec{c}})=\vec{c}]\geq\frac{1}{2}$.
In particular, for every $\vec{c}\in(\parity_d)^k$ there exists a database $D_{\vec{c}}$ labeled by $\vec{c}$ s.t.\ $\Pr_{\AAA}[\AAA(S_{\vec{c}})=\vec{c}]\geq\frac{1}{2}$.
By changing the database $D_{\vec{c}}$ to $D_{\vec{c'}}$ one row at a time while applying the differential privacy constraint, we see that
$$\Pr_{\AAA}[\AAA(D_{\vec{c}})=\vec{c'}]\geq\frac{1}{2}e^{-\epsilon n}.$$
Since the above inequality holds for every two databases $D_{\vec{c}}$ and $D_{\vec{c'}}$, we get
$$\frac{1}{2}\geq\Pr_{\AAA}[\AAA(D_{\vec{c}})\neq\vec{c}]\geq(|\parity_d|^k-1)\frac{1}{2}e^{-\epsilon n}.$$
Solving for $n$, this yields $n=\Omega(kd/\epsilon)$.
\end{proof}

\remove{

\section{discussion}

\begin{itemize}
\item results emphasize importance of semi-supervised learning for differential privacy (general constructions can use a large unlabeled sample and a VC-sized labeled sample).
\item consequences for label privacy? for DP when distribution known?
\item open problems.
\end{itemize}
}

\bibliographystyle{abbrv}
\bibliography{references}

\begin{thebibliography}{10}

\bibitem{AB09}
M.~Anthony and P.~L. Bartlett.
\newblock {\em Neural Network Learning: Theoretical Foundations}.
\newblock Cambridge University Press, New York, NY, USA, 1st edition, 2009.

\bibitem{AnthonySh93}
M.~Anthony and J.~Shawe-Taylor.
\newblock A result of {V}apnik with applications.
\newblock {\em Discrete Applied Mathematics}, 47(3):207--217, 1993.

\bibitem{BST14}
R.~Bassily, A.~Smith, and A.~Thakurta.
\newblock Private empirical risk minimization: Efficient algorithms and tight
  error bounds.
\newblock In {\em FOCS}, pages 464--473, 2014.

\bibitem{BBKN14}
A.~Beimel, H.~Brenner, S.~P. Kasiviswanathan, and K.~Nissim.
\newblock Bounds on the sample complexity for private learning and private data
  release.
\newblock {\em Machine Learning}, 94(3):401--437, 2014.

\bibitem{BKN10}
A.~Beimel, S.~P. Kasiviswanathan, and K.~Nissim.
\newblock Bounds on the sample complexity for private learning and private data
  release.
\newblock In {\em TCC}, pages 437--454, 2010.

\bibitem{BNS13a}
A.~Beimel, K.~Nissim, and U.~Stemmer.
\newblock Characterizing the sample complexity of private learners.
\newblock In {\em ITCS}, pages 97--110, 2013.

\bibitem{BNS13b}
A.~Beimel, K.~Nissim, and U.~Stemmer.
\newblock Private learning and sanitization: Pure vs. approximate differential
  privacy.
\newblock In {\em APPROX-RANDOM}, pages 363--378, 2013.

\bibitem{BNS15}
A.~Beimel, K.~Nissim, and U.~Stemmer.
\newblock Learning privately with labeled and unlabeled examples.
\newblock In {\em SODA}, pages 461--477, 2015.

\bibitem{BDMN05}
A.~Blum, C.~Dwork, F.~McSherry, and K.~Nissim.
\newblock Practical privacy: the {SuLQ} framework.
\newblock In {\em PODS}, pages 128--138, 2005.

\bibitem{BLR08}
A.~Blum, K.~Ligett, and A.~Roth.
\newblock A learning theory approach to noninteractive database privacy.
\newblock {\em J. {ACM}}, 60(2):12, 2013.

\bibitem{BlumerEhHaWa89}
A.~Blumer, A.~Ehrenfeucht, D.~Haussler, and M.~K. Warmuth.
\newblock Learnability and the {V}apnik-{C}hervonenkis dimension.
\newblock {\em J. ACM}, 36(4):929--965, Oct. 1989.

\bibitem{BS98}
D.~Boneh and J.~Shaw.
\newblock Collusion-secure fingerprinting for digital data.
\newblock {\em IEEE Transactions on Information Theory}, 44(5):1897--1905,
  1998.

\bibitem{BNSV15}
M.~Bun, K.~Nissim, U.~Stemmer, and S.~Vadhan.
\newblock Differentially private release and learning of threshold functions.
\newblock In {\em Proceedings of the 56th Annual {IEEE} Symposium on
  Foundations of Computer Science ({FOCS} 2015)}, pages 634--649, Berkeley, CA,
  USA, October 18-20, 2015.

\bibitem{BUV14}
M.~Bun, J.~Ullman, and S.~P. Vadhan.
\newblock Fingerprinting codes and the price of approximate differential
  privacy.
\newblock In {\em STOC}, pages 1--10, 2014.

\bibitem{CH11}
K.~Chaudhuri and D.~Hsu.
\newblock Sample complexity bounds for differentially private learning.
\newblock In S.~M. Kakade and U.~von Luxburg, editors, {\em COLT}, volume~19 of
  {\em JMLR Proceedings}, pages 155--186. JMLR.org, 2011.

\bibitem{DKMMN06}
C.~Dwork, K.~Kenthapadi, F.~McSherry, I.~Mironov, and M.~Naor.
\newblock Our data, ourselves: Privacy via distributed noise generation.
\newblock In {\em EUROCRYPT}, pages 486--503, 2006.

\bibitem{DL09}
C.~Dwork and J.~Lei.
\newblock Differential privacy and robust statistics.
\newblock In {\em Proceedings of the Forty-first Annual ACM Symposium on Theory
  of Computing}, STOC '09, pages 371--380, New York, NY, USA, 2009. ACM.

\bibitem{DMNS06}
C.~Dwork, F.~McSherry, K.~Nissim, and A.~Smith.
\newblock Calibrating noise to sensitivity in private data analysis.
\newblock In {\em TCC}, pages 265--284, 2006.

\bibitem{DRV10}
C.~Dwork, G.~N. Rothblum, and S.~P. Vadhan.
\newblock Boosting and differential privacy.
\newblock In {\em FOCS}, pages 51--60, 2010.

\bibitem{DTTZ14}
C.~Dwork, K.~Talwar, A.~Thakurta, and L.~Zhang.
\newblock Analyze gauss: Optimal bounds for privacy-preserving principal
  component analysis.
\newblock In {\em Proceedings of the 46th Annual ACM Symposium on Theory of
  Computing}, STOC '14, pages 11--20, New York, NY, USA, 2014. ACM.

\bibitem{FFKN09}
D.~Feldman, A.~Fiat, H.~Kaplan, and K.~Nissim.
\newblock Private coresets.
\newblock In M.~Mitzenmacher, editor, {\em Proceedings of the 41st Annual {ACM}
  Symposium on Theory of Computing, {STOC} 2009, Bethesda, MD, USA, May 31 -
  June 2, 2009}, pages 361--370. {ACM}, 2009.

\bibitem{FX14}
V.~Feldman and D.~Xiao.
\newblock Sample complexity bounds on differentially private learning via
  communication complexity.
\newblock In {\em COLT}, pages 1000--1019, 2014.

\bibitem{GHRU11}
A.~Gupta, M.~Hardt, A.~Roth, and J.~Ullman.
\newblock Privately releasing conjunctions and the statistical query barrier.
\newblock In {\em STOC}, pages 803--812, 2011.

\bibitem{HR10}
M.~Hardt and G.~N. Rothblum.
\newblock A multiplicative weights mechanism for privacy-preserving data
  analysis.
\newblock In {\em FOCS}, pages 61--70, 2010.

\bibitem{HT10}
M.~Hardt and K.~Talwar.
\newblock On the geometry of differential privacy.
\newblock In {\em STOC}, pages 705--714, 2010.

\bibitem{KLNRS08}
S.~P. Kasiviswanathan, H.~K. Lee, K.~Nissim, S.~Raskhodnikova, and A.~Smith.
\newblock What can we learn privately?
\newblock {\em {SIAM} J. Comput.}, 40(3):793--826, 2011.

\bibitem{MT07}
F.~McSherry and K.~Talwar.
\newblock Mechanism design via differential privacy.
\newblock In {\em FOCS}, FOCS '07, pages 94--103, Washington, DC, USA, 2007.
  IEEE Computer Society.

\bibitem{NRS07}
K.~Nissim, S.~Raskhodnikova, and A.~Smith.
\newblock Smooth sensitivity and sampling in private data analysis.
\newblock In {\em Proceedings of the 39th Annual {ACM} Symposium on Theory of
  Computing, San Diego, California, USA, June 11-13, 2007}, pages 75--84, 2007.

\bibitem{PSS03}
C.~Peikert, abhi shelat, and A.~Smith.
\newblock Lower bounds for collusion-secure fingerprinting.
\newblock In {\em SODA}, pages 472--479, 2003.

\bibitem{RR10}
A.~Roth and T.~Roughgarden.
\newblock Interactive privacy via the median mechanism.
\newblock In {\em STOC}, pages 765--774, 2010.

\bibitem{ST13}
A.~Smith and A.~Thakurta.
\newblock Differentially private feature selection via stability arguments, and
  the robustness of the lasso.
\newblock In {\em {COLT} 2013}, pages 819--850, 2013.

\bibitem{SU15}
T.~Steinke and J.~Ullman.
\newblock Between pure and approximate differential privacy.
\newblock In {\em TPDP 2015}, 2015.

\bibitem{Tardos08}
G.~Tardos.
\newblock Optimal probabilistic fingerprint codes.
\newblock {\em J. ACM}, 55(2), 2008.

\bibitem{Valiant84}
L.~G. Valiant.
\newblock A theory of the learnable.
\newblock {\em Commun. ACM}, 27(11):1134--1142, Nov. 1984.

\bibitem{Valiant06}
L.~G. Valiant.
\newblock Knowledge infusion.
\newblock In {\em Proceedings, The Twenty-First National Conference on
  Artificial Intelligence and the Eighteenth Innovative Applications of
  Artificial Intelligence Conference, July 16-20, 2006, Boston, Massachusetts,
  {USA}}, pages 1546--1551. {AAAI} Press, 2006.

\end{thebibliography}

\end{document}